\documentclass[runningheads]{llncs}
\usepackage{amssymb, bussproofs, amsmath, hyperref}
\usepackage{color}
\usepackage[all]{xy}
\usepackage{pgf, tikz, color, proof, caption}
\usepackage{changepage, tcolorbox}
\usepackage{relsize, bm, soul}
\usetikzlibrary{decorations.pathreplacing}
\usepackage{enumerate, upgreek}

\usepackage{marvosym}

\definecolor{kgreen}{rgb}{0.09, 0.45, 0.27}
\definecolor{com}{rgb}{0.8, 0.0, 0.0}

\newcommand{\ldml}{\mathsf{G3Ldm}_{n}^{m}}
\newcommand{\calc}{\mathsf{Ldm}_{n}^{m}\mathsf{L}}

\newcommand{\ldmlsa}{\mathsf{Ldm}_{n}^{1}\mathsf{L}}

\newcommand{\stit}{\mathsf{STIT}}

\newcommand{\ldm}{\mathsf{Ldm}_{n}^{m}}
\newcommand{\ldmn}{\mathsf{Ldm}_{n}^{m}}
\newcommand{\ldmp}[2]{\mathsf{Ldm_{#1}^{#2}}}
\newcommand{\ldmsa}{\mathsf{Ldm}_{n}^{1}}

\newcommand{\cut}{(\mathsf{cut})}
\newcommand{\wk}{(\mathsf{wk})}
\newcommand{\ctr}{(\mathsf{ctr})}

\newcommand{\refl}{(\mathsf{refl})}

\newcommand{\stitrefl}{(\mathsf{refl}_{i})}
\newcommand{\stiteucl}{(\mathsf{eucl}_{i})}

\newcommand{\ioa}{(\mathsf{IOA})}
\newcommand{\bridge}{(\mathsf{Bridge})}

\newcommand{\id}{(\mathsf{id})}

\newcommand{\choicerule}{(\mathsf{APC}_{n}^{i})}

\newcommand{\disr}{(\vee)}
\newcommand{\conr}{(\wedge)}
\newcommand{\settr}{(\Box)}
\newcommand{\settdiar}{(\Diamond)}
\newcommand{\stitr}{(\agbox{})}
\newcommand{\stitdiar}{(\agdia)}

\newcommand{\propag}{(\mathsf{Pr}_{i})}

\newcommand{\R}{\mathcal{R}}

\newcommand{\agdia}{\langle i \rangle}
\newcommand{\agdiaone}{\langle 1 \rangle}
\newcommand{\agbox}{[i]}
\newcommand{\agboxone}{[1]}

\newcommand{\lb}{\langle}
\newcommand{\rb}{\rangle}

\newcommand{\apc}{(\mathsf{APC}^i_{n})}

\newtheorem{innercustomlem}{Lemma}
\newenvironment{customlem}[1]
  {\renewcommand\theinnercustomlem{#1}\innercustomlem}
  {\endinnercustomlem}
  
  \newtheorem{innercustomthm}{Theorem}
\newenvironment{customthm}[1]
  {\renewcommand\theinnercustomthm{#1}\innercustomthm}
  {\endinnercustomthm}

\begin{document}

\title{Automating Agential Reasoning: Proof-Calculi and Syntactic Decidability for STIT Logics\thanks{This is a pre-print of an article published in Lecture Notes in Artificial Intelligence. The final authenticated version is available online at: \url{https://doi.org/10.1007/978-3-030-33792-6_13}. Work funded by the projects WWTF MA16-028, FWF I2982 and FWF W1255-N23.}
} 

\author{Tim Lyon\textsuperscript{(\Letter)} \and Kees van Berkel} 

\institute{Institut f\"ur Logic and Computation, Technische Universit\"at Wien, 1040 Wien, Austria  \\ \email{\{lyon,kees\}@logic.at}}




\titlerunning{Automating Agential Reasoning}

\authorrunning{T. Lyon and K. van Berkel}

\maketitle

\begin{abstract} 
This work provides proof-search algorithms and automated counter-model extraction for a class of $\stit$ logics. With this, we answer an open problem concerning syntactic decision procedures and cut-free calculi for $\stit$ logics. A new class of cut-free complete labelled sequent calculi $\ldml$, for multi-agent $\stit$ with at most $n$-many choices, is introduced. We refine the calculi $\ldml$ through the use of propagation rules and demonstrate the admissibility of their structural rules, resulting in the auxiliary calculi $\calc$. In the single-agent case, we show that the refined calculi $\calc$ derive theorems within a restricted class of (forestlike) sequents, allowing us to provide proof-search algorithms that decide single-agent $\stit$ logics. We prove that the proof-search algorithms are correct and terminate. 
\end{abstract}

\section{Introduction}\label{sec_intro}

    Modal logics of $\stit$, an acronym for `seeing to it that', have a long tradition in the formal investigation 
of agency, starting with a series of papers by Belnap and Perloff in the 1980s 
 and culminating in \cite{BelPerXu01}. For the past decades, $\stit$ logic has continued to receive 
 considerable attention, proving itself invaluable in a multitude of fields concerned with formal reasoning about agentive choice making. For example, the framework has been applied to epistemic logic \cite{Bro11}, deontic logic \cite{Hor01,Mur04}, and the formal analysis of legal reasoning \cite{Bro11,LorSar15}. Surprisingly, investigations of the mathematical properties of $\stit$ logics are limited 
  \cite{BalHerTro08,Sch12} and its proof-theory has only been addressed recently 
  \cite{BerLyo19,Wan06}. What is more, despite AI-oriented $\stit$ papers motivating the need of tools for automated reasoning about agentive choice-making \cite{ArkBriBel05,BalHerTro08,BerLyo19}, the envisaged automation results are still lacking.
 The present work will be the first to provide terminating, automated proof-search for a class of $\stit$ logics, including 
 counter-model extraction directly based 
  on failed proof-search. %

The \textit{sequent calculus} 
\cite{Gen35} 
 is an effective framework for proof-search, 
  suitable for automated deduction procedures. Given the metalogical property of \textit{analyticity}, a sequent calculus allows for the construction of proofs by merely decomposing the formula in question. 
In the present work, we employ the \textit{labelled} sequent calculus---a useful formalism for a large class of modal logics \cite{Neg05,Vig00}---and introduce 
 labelled sequent calculi $\ldml$ (with $n,m \in\mathbb{N}$) for multi-agent $\stit$ logics containing 
 limited choice axioms, discussed in \cite{Xu94-2}. 

 In order to appropriate the calculi $\ldml$ for automated proof-search, we take up a \textit{refinement} method presented in \cite{TiuIanGor12}---developed for the more restricted setting of display logic---
and adapt it to the more general setting of labelled 
 calculi. 
  In the refinement process the \textit{external} character of labelled systems---namely, the explicit presence of the semantic structure---is made \textit{internal} through the use of alternative, yet equivalent, \textit{propagation rules} \cite{TiuIanGor12}. The tailored propagation rules restrict and simplify 
the sequential structures needed in derivations, producing, for example, shorter proofs. 
  Moreover, 
  one can show that through the use of propagation rules, the 
  structural rules, capturing the behavior of the logic's modal operators, are admissible. In our case, the resulting refined calculi $\calc$ 
  derive theorems using only 
   forestlike sequents, allowing us to adapt methods from \cite{TiuIanGor12} and provide 
  correct and terminating proof-search algorithms for this class of $\stit$ logics. 

In short, the contribution of this paper will be threefold: First, we provide new sound and cut-free complete labelled sequent calculi $\ldml$ for all multi-agent $\stit$ logics $\ldm$ (with $n,m\in\mathbb{N}$) discussed in \cite{Xu94-2}---thus extending the class of logics addressed in \cite{BerLyo19}. Second, we show how to refine these calculi to obtain new calculi $\calc$, which are suitable for proof-search. Last, for each $n \in \mathbb{N}$, we provide a terminating proof-search algorithm deciding the single-agent $\stit$ logic $\ldmsa$. Although \cite{GroLorSch15} provides a polynomial reduction of $\ldmn$ into the modal logic $\mathsf{S5}$ (providing decidability via $\mathsf{S5}$-SAT), the present work has the advantage that it offers a syntactic decision procedure within the unreduced $\ldmn$ language and is modular, that is, it will allow us to extend our work to a variety of $\stit$ logics. We conclude by discussing the prospects of generalizing the latter results to the multi-agent setting.

 

The paper is structured as follows: We start by introducing the class of logics $\ldm$ in Sec.~\ref{sec_ldm}. In Sec.~\ref{sec_ldml}, corresponding labelled calculi $\ldml$ are provided, which 
 will subsequently be refined, resulting in the calculi $\calc$. 
 We devote Sec.~\ref{sec_proofsearch} to proof-search algorithms and counter-model extraction. 
    
\section{Logical Preliminaries}\label{sec_ldm}

$\stit$ logic refers to 
 a group of modal logics using operators that capture agential choice-making. 
  The $\stit$ logics $\ldmn$, 
 which will be considered throughout this paper, employ two types of modal operators: First, they contain a \textit{settledness} operator $\Box$ expressing which formulae are `settled true' at a current moment. Second, they contain, for each agent $i$ in the language, an atemporal---i.e., instantaneous---\textit{choice} operator $[i]$ expressing that `agent $i$ sees to it that'. This basic choice operator is referred to as the \textit{Chellas} $\stit$ \cite{BelPerXu01}. Using both operators, 
 one can define the more refined notion of \textit{deliberative} $\stit$: 
  i.e., $[i]_d\phi$ iff $[i]\phi\land\lnot \Box \phi$.  Intuitively, 
 $[i]_d\phi$ holds when `agent $i$ sees to it that $\phi$ and it is possible for $\phi$ not to hold'. 
 The multi-agent language for $\ldmn$ is defined accordingly: 

\begin{definition}[The Language $\mathcal{L}^m$~\cite{HerSch08}]\label{def:ldm-language} Let $Ag = \{1,2,...,m\}$ be a finite set of agent labels 
and let $Var =\{p,q,r...\}$ be a countable set of propositional variables. 
 $\mathcal{L}^m$ is defined via the following BNF grammar:
\begin{center}
$\phi ::= p \ | \ \overline{p} \ | \ (\phi \wedge \phi) \ | \ (\phi\vee \phi) \ | \ (\Box \phi) \ | \ (\Diamond \phi) \ | \ ([i] \phi) \ | \ (\lb i \rb \phi)$
\end{center}
where $i\in Ag$ and $p \in Var$.
\end{definition}

Notice, the language $\mathcal{L}^m$ consists of formulae in negation normal form. This notation allows us to reduce the number of rules in our calculi, enhancing the readability and simplicity of our proof theory. The negation of $\phi$, written as $\overline{\phi}$, is obtained by replacing each operator with its dual, each positive atom $p$ with its negation $\overline{p}$, and each $\overline{p}$ with its positive variant $p$~\cite{BerLyo19}. Consequently, we obtain the following abbreviations: $\phi\rightarrow\psi \textit{ iff } \overline{\phi} \lor \psi$, $\phi\leftrightarrow \psi \textit{ iff } (\phi\rightarrow\psi)\land(\psi\rightarrow\phi)$, $\top \textit{ iff } p\lor\overline{p}$, and $\bot \textit{ iff } p\land\overline{p}$. We will freely use these abbreviations throughout this paper. Since we are working in negation normal form, diamond-modalities are introduced as separate primitive operators. We take $\agdia$ and $\Diamond$ as the duals of $[i]$ and $\Box$, respectively. 

Following \cite{HerSch08}, since we work with instantaneous, atemporal $\stit$ 
 it suffices to regard only single choice-moments in our relational frames. This means that we can forgo the traditional branching time structures of basic, atemporal $\stit$ logic~\cite{BelPerXu01}. In what follows, we define $\ldm$ frames as those $\stit$ frames in which $n > 0$ limits the amount of choices available to each agent to at most $n$-many choices (imposing no limitation when $n = 0$).\footnote{For a discussion of the philosophical utility of reasoning with 
 limited choice see \cite{Xu94-2}.}

\begin{definition}[Relational $\ldm$ Frames and Models]\label{def:models-and-frames} Let $|Ag|=m$ and let $\R_{i}(w) := \{v\in W \ | \ (w,v)\in \R_{i}\}$ for $i \in  Ag$. An \emph{$\ldmn$-frame} is defined as a tuple $F = (W, \{\R_{i} \ | \ i \in Ag\})$ where $W\neq\emptyset$ is a set of worlds $w,v,u...$ and:

\begin{itemize}

\item[{\rm \textbf{(C1)}}] For each $i\in Ag$, $\R_{i} \subseteq W\times W$ is an equivalence relation;

\item[{\rm \textbf{(C2)}}] For all $u_{1},...,u_{m} \in W$, $\bigcap_{i} \R_{i}(u_{i}) \neq \emptyset$;

\item[{\rm \textbf{(C3)}}] Let $n > 0$ and $i \in Ag$, then 
\vspace{-5pt}
\begin{center}
For all $w_{0}, w_{1}, \cdots, w_{n} \in W$, $\displaystyle{ \bigvee_{0 \leq k \leq n-1 \text{, } k+1 \leq j \leq n} \R_{i} w_{k}w_{j}} $
\end{center}

\end{itemize}
An $\ldm$-model is a tuple $M = (F,V)$ where $F$ is an $\ldm$-frame and $V$ is a valuation 
 assigning propositional variables to subsets of $W$, i.e. $V{:}\ Var \mapsto \mathcal{P}(W)$. Additionally, we stipulate that condition \textbf{(C3)} is omitted when $n = 0$.
\end{definition}

As in~\cite{HerSch08}, the set of worlds $W$ is taken to represent a single moment in which agents from $Ag$ are making their decision. Following \textbf{(C1)}, for every agent $i$, the relation $\R_i$ is an equivalence relation; that is, $\R_i$ functions as a partitioning of $W$ into what will be called \textit{choice-cells} for agent $i$. Each choice-cell represents a set of possible worlds that may be realized by a choice of the agent. The condition \textbf{(C2)} expresses the $\stit$ principle \textit{independence of agents}, ensuring that any combination of choices, available to different agents, is consistent. The last condition \textbf{(C3)}, represents the $\stit$ principle which limits the amount of choices available to an agent to a maximum of $n$.  
 For a philosophical discussion of these principles we refer to \cite[Ch.~7C]{BelPerXu01}.

\begin{definition}[Semantic Clauses for $\mathcal{L}^m$~\cite{BerLyo19,HerSch08}]\label{Semantics_ldm} Let $M$ be an $\ldm$-model $(W, \{\R_{i} \ | \ i \in Ag\}, V)$ and let $w$ be a world in its domain $W$. The \emph{satisfaction} of a formula $\phi\in\mathcal{L}^m$ on $M$ at $w$ is inductively defined as follows:
\begin{small}
\vspace{-1pt}
\addtolength{\linewidth}{1.5em}
\begin{multicols}{2}
\begin{itemize}
\item[1.] $M, w \Vdash p$ iff $w \in V(p)$

\item[2.] $M, w \Vdash \overline{p}$ iff $w \not\in V(p)$

\item[3.] $M, w \Vdash \phi \wedge \psi$ iff $M, w \Vdash \phi$ and $M, w \Vdash \psi$

\item[4.] $M, w \Vdash \phi \vee \psi$ iff $M, w \Vdash \phi$ or $M, w \Vdash \psi$

\end{itemize}
\columnbreak
\begin{itemize}
\item[5.] $M, w \Vdash \Box \phi$ iff $\forall u \in W$, $ M, u \Vdash \phi$

\item[6.] $M, w \Vdash \Diamond \phi$ iff $\exists u \in W$, $M, u \Vdash \phi$

\item[7.] $M, w \Vdash [i] \phi$ iff $\forall u \in \R_{i}(w)$, $M, u \Vdash \phi$

\item[8.] $M, w \Vdash \lb i \rb \phi$ iff $\exists u \in \R_{i}(w)$, $M, u \Vdash \phi$
\end{itemize}
\end{multicols}

\end{small}
\vspace{-1pt}
\noindent A formula $\phi$ is \emph{globally true} on $M$ (\textit{i.e.} $M \Vdash \phi$) iff it is satisfied at every world $w$ in the domain $W$ of $M$. A formula $\phi$ is \emph{valid} (\textit{i.e.} $\Vdash \phi$) iff it is globally true on every $\ldm$-model. Last, $\Gamma$ semantically implies $\phi$, written $\Gamma \Vdash \phi$, iff for all models $M$ and worlds $w$ of $W$ in $M$, if $M,w \Vdash \psi$ for all $\psi \in \Gamma$, then $M,w \Vdash \phi$.

\end{definition}
It is worth emphasizing 
 that the semantic interpretation of 
$\Box$ 
 refers to the domain of the model in its entirety; i.e., $\phi$ is settled true iff $\phi$ is globally true. This is an immediate consequence of considering instantaneous $\stit$ in a single-moment setting (cf. semantics where a relation $\R_{\Box}$ is 
 introduced for $\Box$, e.g., \cite{BerLyo19}). 

The Hilbert calculus for $\ldm$ in Fig. 1 is taken from \cite{Xu94-2}. Apart from the propositional axioms, it consists of $\mathsf{S5}$ axiomatizations for $\Box$ and $[i]$, for each $i{\in} Ag$. It contains the standard bridge axiom $\bridge$, linking $[i]$ to $\Box$. Furthermore, it contains an independence of agents axiom $\ioa$, as well as an $n$-choice axiom $\apc$ for each $i{\in} Ag$. The rules are modus ponens and $\Box$-necessitation. 
\begin{figure}[t]\label{fig:hilbert_ldm}
\noindent\hrule
\begin{center}
\begin{tabular}{c @{\hskip 1em} c @{\hskip 1em} c @{\hskip 1em} c}

$\phi \rightarrow (\psi \rightarrow \phi)$

&

$(\overline{\psi} \rightarrow \overline{\phi}) \rightarrow (\phi \rightarrow \psi)$

&

$(\phi \rightarrow (\psi \rightarrow \chi)) \rightarrow ((\phi \rightarrow \psi) \rightarrow (\phi \rightarrow \chi))$

\end{tabular}

\ \\

\begin{tabular}{c @{\hskip 1.5em} c @{\hskip 1.5em} c @{\hskip 1.5em} c @{\hskip 1.5em} c }

$(\mathsf{S5}\Box)\;\; \Box (\phi \rightarrow \psi) \rightarrow (\Box \phi \rightarrow \Box \psi)$

&

$\Box \phi \rightarrow \phi$

&

$\Diamond \phi \rightarrow \Box \Diamond \phi$
&

$\Box \phi \vee \Diamond \overline{\phi}$

&

\end{tabular}
\end{center}

\begin{center}
\begin{tabular}{c @{\hskip 1.5em} c @{\hskip 1.5em} c @{\hskip 1.5em} c @{\hskip 1.5em} }

$(\mathsf{S5}[i])\;\;  \agbox{} (\phi \rightarrow \psi) \rightarrow (\agbox{} \phi \rightarrow \agbox{} \psi)$

&

$\agbox{} \phi \rightarrow \phi$

&

$\lb i \rb \phi \rightarrow \agbox{} \lb i \rb \phi$

&

$[i] \phi \vee \lb i \rb \overline{\phi}$

\end{tabular}
\end{center}

\begin{center}
\begin{tabular}{c @{\hskip 1em} c @{\hskip 1em} c @{\hskip 1em} c @{\hskip 1em} c}
$\ioa \; \bigwedge_{i \in Ag} \Diamond [i] \phi_{i} \rightarrow \Diamond ( \bigwedge_{i \in Ag}[i] \phi_{i})$

&

$\bridge \; \Box \phi \rightarrow [i] \phi$

&

\AxiomC{$\phi$}
\UnaryInfC{$\Box \phi$}
\DisplayProof

&

\AxiomC{$\phi$}
\AxiomC{$\phi \rightarrow \psi$}
\BinaryInfC{$\psi$}
\DisplayProof

\end{tabular}
\end{center}

\begin{center}
\begin{tabular}{c}
$\apc \; \Diamond [i] \phi_{1} \land \Diamond (\overline{\phi}_{1} \land [i] \phi_{2}) \land \cdots \land \Diamond (\overline{\phi}_{1} \land \cdots \land \overline{\phi}_{n-1} \land [i] \phi_{n}) \rightarrow \phi_{1} \lor \cdots \lor \phi_{n}$
\end{tabular}
\end{center}


\hrule
\caption{The Hilbert calculus for $\ldm$~\cite{BelPerXu01,Xu94-2}. A \emph{derivation of $\phi$ in $\ldm$} from a set of premises $\Gamma$, is written as $\Gamma \vdash_{\ldm} \phi$, and is defined inductively in the usual way. 
When $\Gamma$ is the empty set, we refer to $\phi$ as 
a \emph{theorem} and write $\vdash_{\ldm} \phi$.}
\vspace{-3pt}
\end{figure}


\begin{theorem}[Soundness and Completeness~\cite{HerSch08,Xu94-2}]
For any formula $\phi\in\mathcal{L}^m$, $\Gamma \vdash_{\ldm} \phi$ if and only if $ \Gamma \Vdash \phi$.
\end{theorem}

    
\section{Refinement of the Calculi $\ldml$}\label{sec_ldml}

    In this section, we introduce the labelled calculi $\ldml$ for multi-agent $\stit$ logics (with limited choice). Our calculi are modified, extended versions of the labelled calculi for the logics $\mathsf{Ldm}_{0}^{m}$ (with $m \in \mathbb{N}$) proposed in \cite{BerLyo19} and cover a larger class of logics. The calculi $\ldml$ possess fundamental proof-theoretic properties such as contraction- and cut-admissibility which 
  follow from the general results on labelled calculi established in~\cite{Neg05}.
 The main goal of this section is to refine the $\ldml$ calculi through the elimination of structural rules, resulting in new calculi $\calc$ that derive theorems within a restricted class of sequents. As a result of adopting the approach in \cite{HerSch08}, 
 the omission of the relational structure corresponding to the $\Box$ modality 
 offers a simpler 
 approach to proving the admissibility of structural rules in the the presence of propagation rules (Sec.~\ref{sec:refinement}). Let us start by introducing the class of $\ldml$ calculi.


\subsection{The $\ldml$ Calculi}

We define labelled sequents $\Lambda$ via the following BNF grammar:
\begin{center}
$\Lambda ::= x:\phi \ | \ \Lambda, \Lambda \ | \ \R_{i}xy, \Lambda$
\end{center}
where $i \in Ag$, $\phi \in \mathcal{L}^{m}$ and $x,y$ are from a denumerable set of labels $Lab = \{ x, y, z, ... \}$. Labelled sequents consist exclusively of labelled formulae of the form $x:\phi$ and relational atoms of the form $\R_{i}xy$. For this reason, sequents can be partitioned into two parts: 
we sometimes use the notation $\R,\Gamma$ to denote labelled sequents, where $\R$ is the part consisting of relational atoms and $\Gamma$ is the part consisting of labelled formulae. Last, we interpret the commas between relational atoms in $\R$ conjunctively, the comma between $\R$ and $\Gamma$ in $\R, \Gamma$ implicationally, and the commas between labelled formulae in $\Gamma$ disjunctively (cf. Def.~\ref{def:sequent-satisfaction-validity}).

The 
labelled $\stit$ calculi $\ldml$ (where $n,m\in\mathbb{N}$) 
 are shown in Fig. 2. 
  Note that for each agent $i \in Ag$, we obtain a copy for each of the rules $\stitdiar$, $\stitr$, $\stitrefl$, $\stiteucl$, and $(\mathsf{APC}^{i}_{n})$. We refer to $\stitrefl$, $\stiteucl$, $\ioa$, and $(\mathsf{APC}^{i}_{n})$ as the \emph{structural rules} of $\ldml$. 
 The rule $\ioa$ captures the \textit{independence of agents} principle. Furthermore, the rule schema $\choicerule$, limiting the amount of choices available to agent $i$, provides different rules depending on the value of $n$ in $\ldml$ (we reserve 
$n=0$ 
 to assert that the rule does not appear). 
 When $n > 0$, the $(\mathsf{APC}^{i}_{n})$ rule contains $n(n+1)/2$ premises, where each sequent $\R, \R_{i}x_{k}x_{j}, \Gamma$ (for $0 \leq k \leq n-1$ and $k+1 \leq j \leq n$) represents a different premise of the rule. As an example, for $n=1$ and $n=2$ the rules for agent $i$ are:

\begin{small}
\begin{center}
\begin{tabular}{
 c @{\hskip 1em} c}
\AxiomC{$\R, \R_{i}w_{0}w_{1}, \Gamma$}
\RightLabel{$(\mathsf{APC}^{i}_{1})$}
\UnaryInfC{$\R,\Gamma$}
\DisplayProof

&
\def\defaultHypSeparation{\hskip .1in}
\AxiomC{$\R, \R_{i}w_{0}w_{1}, \Gamma$}
\AxiomC{$\R, \R_{i}w_{0}w_{2}, \Gamma$}
\AxiomC{$\R, \R_{i}w_{1}w_{2}, \Gamma$}
\RightLabel{$(\mathsf{APC}^{i}_{2})$}
\TrinaryInfC{$\R, \Gamma$}
\DisplayProof
\end{tabular}
\end{center}
\end{small}

\begin{figure}[t]\label{fig:g3ldmn}
\noindent\hrule
\begin{center}
\begin{tabular}{c c}

\AxiomC{ }
\RightLabel{$\id$}
\UnaryInfC{$\R, w:p, w:\overline{p}, \Gamma$}
\DisplayProof

&

\AxiomC{$\R, w: \phi \wedge \psi, w: \phi, \Gamma$}
\AxiomC{$\R, w: \phi \wedge \psi, w: \psi, \Gamma$}
\RightLabel{$\conr$}
\BinaryInfC{$\R, w: \phi \wedge \psi, \Gamma$}
\DisplayProof

\end{tabular}
\end{center}

\begin{center}
\begin{tabular}{c c}
\AxiomC{$\R, w: \phi \vee \psi, w: \phi, w : \psi, \Gamma$}
\RightLabel{$\disr$}
\UnaryInfC{$\R, w: \phi \vee \psi, \Gamma$}
\DisplayProof

&

\AxiomC{$\R, \R_{i}wv, v: \phi, \Gamma$}
\RightLabel{$\stitr^{\dag}$}
\UnaryInfC{$\R, w: [i] \phi, \Gamma$}
\DisplayProof

\end{tabular}
\end{center}

\begin{center}
\begin{tabular}{c c c}

\AxiomC{$\R, w: \Box \phi, v: \phi, \Gamma$}
\RightLabel{$\settr^{\dag}$}
\UnaryInfC{$\R, w: \Box \phi, \Gamma$}
\DisplayProof

&

\AxiomC{$\R, w: \Diamond \phi, u: \phi, \Gamma$}
\RightLabel{$\settdiar$}
\UnaryInfC{$\R, w: \Diamond \phi, \Gamma$}
\DisplayProof

&

\AxiomC{$\R,\R_{1}u_{1}v, ..., \R_{m}u_{m}v, \Gamma$}
\RightLabel{$\ioa^{\dag}$}
\UnaryInfC{$\R,\Gamma$}
\DisplayProof

\end{tabular}
\end{center}

\begin{center}
\begin{tabular}{c c c}

\AxiomC{$\R, \R_{i}wu, w: \lb i \rb \phi, u: \phi, \Gamma$}
\RightLabel{$\stitdiar$}
\UnaryInfC{$\R, \R_{i}wu, w: \lb i \rb \phi, \Gamma$}
\DisplayProof

&

\AxiomC{$\R, \R_{i}ww, \Gamma$}
\RightLabel{$\stitrefl$}
\UnaryInfC{$\R, \Gamma$}
\DisplayProof

&

\AxiomC{$\R, \R_{i}wu, \R_{i}wv, \R_{i}uv, \Gamma$}
\RightLabel{$\stiteucl$}
\UnaryInfC{$\R, \R_{i}wu, \R_{i}wv, \Gamma$}
\DisplayProof

\end{tabular}
\end{center}

\begin{center}
\begin{tabular}{c}
\AxiomC{$\Big\{ \R, \R_{i}w_{k}w_{j}, \Gamma \ \Big| \ 0 \leq k \leq n-1 \text{, } k+1 \leq j \leq n \Big\}$}
\RightLabel{$\choicerule$}
\UnaryInfC{$\R,\Gamma$}
\DisplayProof
\end{tabular}
\end{center}
\hrule
\caption{The $\ldml$ labelled calculi. The superscript $\dag$ on the $\settr$, $\stitr$, and $\ioa$ rule names indicates an eigenvariable condition: the variable $v$ occurring in the premise of the rule cannot occur in the context of the premise (or, equivalently, in the conclusion).}
\end{figure}


\begin{theorem}\label{thm:g3ldm-properties} The $\ldml$ calculi have the following properties:

\begin{enumerate}

\item  All sequents of the form $\R, w : \phi, w: \overline{\phi}, \Gamma$ are derivable;

\item Variable-substitution is height-preserving admissible;

\item All inference rules are height-preserving invertible;

\item Weakening and contractions are height-preserving admissible:

\begin{center}
\begin{tabular}{c @{\hskip 1em} c @{\hskip 1em} c}
\AxiomC{$\R,\Gamma$}
\RightLabel{$\wk$}
\UnaryInfC{$\R,\R',\Gamma',\Gamma$}
\DisplayProof

&

\AxiomC{$\R,\R',\R',\Gamma$}
\RightLabel{$\ctr_{\mathsf{R}}$}
\UnaryInfC{$\R,\R',\Gamma$}
\DisplayProof

&

\AxiomC{$\R,\Gamma',\Gamma',\Gamma$}
\RightLabel{$\ctr_{\mathsf{F}}$}
\UnaryInfC{$\R,\Gamma',\Gamma$}
\DisplayProof
\end{tabular}
\end{center}

\item The cut rule is admissible:

\begin{center}
\begin{tabular}{c}
\AxiomC{$\R,x:\phi,\Gamma$}
\AxiomC{$\R,x:\overline{\phi},\Gamma$}
\RightLabel{$\cut$}
\BinaryInfC{$\R,\Gamma$}
\DisplayProof
\end{tabular}
\end{center}

\item For every formula $\phi \in \mathcal{L}^m$, $w: \phi$ is derivable in $\ldml$ if and only if $\vdash_{\ldm} \phi$, i.e., $\ldml$ is sound and complete relative to $ \ldm$.

\end{enumerate}
\end{theorem}

\begin{proof} The proof is a basic adaption of \cite{Neg05} and can be found in App.~\ref{appendix}.
\qed
\end{proof}

Proof-theoretic properties like 
 those expressed in (4) and (5) of Thm. \ref{thm:g3ldm-properties} are essential 
 when designing decidability procedures via proof-search. In constructing a proof of a sequent, proof-search algorithms proceed by applying inference rules of a calculus bottom-up. A bottom-up application of the $\cut$ rule in a proof-search procedure, however, requires one to guess the \emph{cut formula} $\phi$, and thus risks non-termination in the algorithm. (One can think of similar arguments 
 why 
 $\ctr_{\mathsf{R}}$ and $\ctr_{\mathsf{F}}$ risk non-termination.) It is thus crucial that such rules are \emph{admissible}; \textit{i.e.} everything derivable with these rules, is derivable without them. 

\begin{remark}\label{rem:contraction} To obtain contraction 
 admissibility (Thm. \ref{thm:g3ldm-properties}-(4)) labelled calculi must satisfy the \emph{closure condition}~\cite{Neg05}: if a substitution of variables in a structural rule brings about a duplication of relational atoms in the conclusion, then the calculus must contain another instance of the rule with this duplication contracted. 
 
 We observe that if we substitute the variable $u$ for $v$ in the structural rule $\stiteucl$ (below left), we obtain the rule $\stiteucl^{*}$ (below right), when the atom $\R_{i}wu$ is contracted:
\begin{center}
\begin{tabular}{c @{\hskip 2em} c}
\AxiomC{$\R, \R_{i}wu, \R_{i}wu, \R_{i}uu, \Gamma$}
\RightLabel{$\stiteucl$}
\UnaryInfC{$\R, \R_{i}wu, \R_{i}wu, \Gamma$}
\DisplayProof

&

\AxiomC{$\R, \R_{i}wu, \R_{i}uu, \Gamma$}
\RightLabel{$\stiteucl^{*}$}
\UnaryInfC{$\R, \R_{i}wu, \Gamma$}
\DisplayProof
\end{tabular}
\end{center}
Thus, following the closure condition, we must also add $\stiteucl^{*}$ to our calculus. However, $\stiteucl^{*}$ is a special instance of the $\stitrefl$ rule, and hence it is admissible; therefore, we can omit its inclusion in our calculi. None of the other structural rules possess duplicate relational atoms in their conclusions under a substitution of variables, and so, each $\ldml$ calculus satisfies the closure condition.
\end{remark}

\subsection{Extracting the $\calc$ Calculi}\label{sec:refinement}

We now refine the $\ldml$ calculi, extracting new
  $\calc$ calculi to which proof-search techniques from~\cite{TiuIanGor12} may be adapted. In short, we introduce new rules to our calculi, called \emph{propagation rules}, which are well-suited 
  for proof-search and imply the admissibility of the less suitable structural rules $\stitrefl$ and $\stiteucl$. 
   
Propagation rules are special sequent rules that possess a nonstandard side condition, consisting of two components. For the first component (1), we transform the sequent occurring in the premise of the rule into an \emph{automaton}. The labels appearing in the sequent determine the states of the automaton, whereas the relational atoms of the sequent determine the transitions between these states. The following definition, based on 
\cite[Def. 4.1]{TiuIanGor12}, makes this notion precise:

\begin{definition}[Propagation Automaton]\label{def:propagation-automaton} Let $\Lambda$ 
 be a labelled sequent, $Lab(\Lambda)$ be the set of labels occurring in $\Lambda$, and $w, u \in Lab(\Lambda)$. We define a \emph{propagation automaton} $\mathcal{P}_{\Lambda}(w,u)$ to be the tuple $(\Sigma,S,I,F,\delta)$ s.t. (i) $\Sigma := \{\agdia \ | \ i \in Ag\}$ is the automaton's \emph{alphabet}, (ii) $S := Lab(\Lambda)$ is the \emph{set of states}, (iii) $I := \{w\}$ is the \emph{initial state}, (iv) $F := \{u\}$ is the \emph{accepting state}, and (v) $\delta : S \times \Sigma \to S$ is the \emph{transition function} where $\delta(v,\agdia) = v'$ and $\delta(v',\agdia) = v$ iff $\R_{i}vv' \in \Lambda$. 

We will often write $v \overset{\agdia}{\longrightarrow} v'$ instead of $\delta(v,\agdia) = v'$ to denote
 a transition between states. A \emph{string} is a, possibly empty, concatenation of symbols from $\Sigma$ (where $\varepsilon$ indicates the empty string). We say that an automaton \emph{accepts} a string $\omega = \langle i_{1} \rangle \langle i_{2} \rangle \cdots \langle i_{k} \rangle$ iff there exists a transition sequence 
  $w \overset{\langle i_{1} \rangle}{\longrightarrow} v \overset{\langle i_{2} \rangle}{\longrightarrow} \cdots \overset{\langle i_{k} \rangle}{\longrightarrow} u$ from the initial state $w$ to the accepting state $u$. 
  Last, we will abuse notation and use $\mathcal{P}_{\Lambda}(w,u)$ equivocally to represent both the automaton and the set of strings $\omega$ accepted by the automaton, \textit{i.e.} $\{\omega \ | \ \mathcal{P}_{\Lambda}(w,u) \text{ accepts string } \omega \}$. 
The use of notation can be determined from the context.

\end{definition}

The second component (2) of the rule's side condition 
 restricts the application of the rule to a particular language that specifies and determines which types of strings occurring in the automaton allow for a correct application of the propagation rule. 
 We define this 
  language accordingly:

\begin{definition}[Agent $i$ 
 Application Language]\label{def:appl_lang} For each $i \in Ag$, we define 
 the \emph{application language} $L_{i}$ to be the language generated from the regular expression $\langle i \rangle^{*}$, that is, $L_{i} = \{\varepsilon, \agdia, \agdia \agdia, \agdia \agdia \agdia, \cdots \}$ with $\varepsilon$ the empty string.\footnote{For further information on regular languages and expressions, consult~\cite{Sip06}.} 
\end{definition}

Bringing components (1) and (2) together, a propagation rule is applicable only if the associated propagation automaton accepts a certain string---corresponding to a path of relational atoms in the premise of the rule---and the string is in the application language.

\begin{definition}[Propagation Rule]\label{def:propagation-rules} Let $i \in Ag$, $\Lambda_{1} = \R, w: \agdia \phi, u:\phi, \Gamma$, and $\Lambda_{2} = \R, w: \agdia \phi, \Gamma$. The \emph{propagation rule} $\propag$ 
is defined as follows:

\begin{center}
\begin{tabular}{c}
\AxiomC{$\R, w: \agdia \phi, u:\phi, \Gamma$}
\RightLabel{$\propag^{\dag\dag}$}
\UnaryInfC{$\R, w: \agdia \phi, \Gamma$}
\DisplayProof
\end{tabular}
\end{center}
The superscript $\dag\dag$ indicates that $\mathcal{P}_{\Lambda_{k}}(w,u) \cap L_{i} \neq \emptyset$ for $k\in\{1,2\}$.\footnote{Observe that $\mathcal{P}_{\Lambda_{1}}(w,u) {=} \mathcal{P}_{\Lambda_{2}}(w,u)$. Hence, deciding 
 which automaton to employ in determining the side condition is inconsequential: when applying the rule top-down we may consult $\Lambda_1$, whereas during bottom-up proof-search we may regard $\Lambda_2$.}

We use $\mathsf{PR} := \{\propag \ | \ i \in Ag\}$ to represent the set of all propagation rules. 
\end{definition}

The underlying intuition of the rule (applied bottom-up) is that, given some labelled sequent $\Lambda$, a formula $\phi$ is propagated from $w: \agdia\phi$ to another label $u$,
 if $w$ and $u$ are connected 
 by a sequence of $\R_{i}$ relational atoms in $\Lambda$ (with $i$ fixed). In the corresponding propagation automaton $\mathcal{P}_{\Lambda}(w,u)$, this amounts to the existence of a string $\omega\in \mathcal{P}_{\Lambda}(w,u)\cap L_i$ which represents a sequence of transitions from $w$ to $u$, such that all transitions are solely labelled with $\agdia$. 
  To see how the language $L_i$ secures the soundness of the rule, we refer to Thm. \ref{thm:soundness-ldmnl}. For an introduction to propagation rules and propagation automata, see 
  \cite{TiuIanGor12}.

Let us make the introduced notions more concrete by providing an example:

\begin{example}\label{eg1:propagation-application} Let $\Lambda = \R_{1}wu, \R_{2}uv, \R_{1}vz, w : \langle 1 \rangle \phi$. The propagation automaton $\mathcal{P}_{\Lambda}(w,z)$ is depicted graphically as (where the single-boxed node $w$ designates the initial state and a double-boxed node $z$ represents the accepting state):
\begin{center}
\begin{tabular}{c}
\xymatrix{
\boxed{w}\ar@/^-1pc/@{.>}[rr]|-{\langle 1 \rangle} &   &   u\ar@/^-1pc/@{.>}[rr]|-{\langle 2 \rangle}\ar@/^-1pc/@{.>}[ll]|-{\langle 1 \rangle}   &  &   v \ar@/^-1pc/@{.>}[rr]|-{\langle 1 \rangle}\ar@/^-1pc/@{.>}[ll]|-{\langle 2 \rangle} &  &  \boxed{\boxed{z}}\ar@/^-1pc/@{.>}[ll]|-{\langle 1 \rangle}
}
\end{tabular}
\end{center}

Observe that every string the automaton accepts must contain at least one $\langle 2 \rangle$ symbol. Since no string of this form exists in $L_{1}$, it is not valid to propagate the formula $\phi$ to $z$. That is, the sequent $\R_{1}wu, \R_{2}uv, \R_{1}vz, w : \langle 1 \rangle \phi, z:\phi$ does not follow from applying the propagation rule $(\mathsf{Pr}_1)$ (bottom-up) to $\Lambda$.

On the other hand, consider the propagation automaton $\mathcal{P}_{\Lambda}(w,u)$:
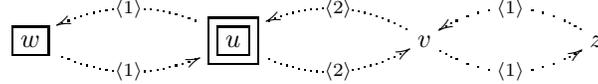
\begin{center}
\begin{tabular}{c}
\xymatrix{
\boxed{w}\ar@/^-1pc/@{.>}[rr]|-{\langle 1 \rangle} &   &   \boxed{\boxed{u}}\ar@/^-1pc/@{.>}[rr]|-{\langle 2 \rangle}\ar@/^-1pc/@{.>}[ll]|-{\langle 1 \rangle}   &  &   v \ar@/^-1pc/@{.>}[rr]|-{\langle 1 \rangle}\ar@/^-1pc/@{.>}[ll]|-{\langle 2 \rangle} &  &  z\ar@/^-1pc/@{.>}[ll]|-{\langle 1 \rangle}
}
\end{tabular}
\end{center}
The automaton accepts the simple string $\agdiaone$, which is included in the language $L_{1}$. Therefore, it is permissible to apply the propagation rule $(\mathsf{Pr}_1)$ (bottom-up) and derive $\R_{1}wu, \R_{2}uv, \R_{1}vz, w : \langle 1 \rangle \phi, u:\phi$ from $\Lambda$.

\end{example}

\begin{remark} We observe that both of the languages $\mathcal{P}_{\Lambda}(w,u)$ and $L_{i}$ are regular, and thus, the problem of determining whether $\mathcal{P}_{\Lambda}(w,u) \cap L_{i} \neq \emptyset$, 
 is decidable~\cite{TiuIanGor12}. Consequently, the propagation rules in $\mathsf{PR}$ may be integrated into our proof-search algorithm without risking non-termination. 
\end{remark}

\begin{figure}[t]\label{fig:ldmnl}
\noindent\hrule
\begin{center}
\begin{tabular}{c @{\hskip 2em} c}

\AxiomC{ }
\RightLabel{$\id$}
\UnaryInfC{$\R, w:p, w:\overline{p}, \Gamma$}
\DisplayProof

&

\AxiomC{$\R, w: \phi \wedge \psi, w: \phi, \Gamma$}
\AxiomC{$\R, w: \phi \wedge \psi, w: \psi, \Gamma$}
\RightLabel{$\conr$}
\BinaryInfC{$\R, w: \phi \wedge \psi, \Gamma$}
\DisplayProof

\end{tabular}
\end{center}

\begin{center}
\begin{tabular}{c @{\hskip 1em} c @{\hskip 1em} c}
\AxiomC{$\R, w: \phi \vee \psi, w: \phi, w : \psi, \Gamma$}
\RightLabel{$\disr$}
\UnaryInfC{$\R, w: \phi \vee \psi, \Gamma$}
\DisplayProof

&

\AxiomC{$\R, w: \Box \phi, v: \phi, \Gamma$}
\RightLabel{$\settr^{\dag}$}
\UnaryInfC{$\R, w: \Box \phi, \Gamma$}
\DisplayProof

&

\AxiomC{$\R, w: \Diamond \phi, u: \phi, \Gamma$}
\RightLabel{$\settdiar$}
\UnaryInfC{$\R, w: \Diamond \phi, \Gamma$}
\DisplayProof
\end{tabular}
\end{center}

\begin{center}
\begin{tabular}{c @{\hskip 1em} c}

\AxiomC{$\R,\R_{1}u_{1}v, ..., \R_{m}u_{m}v, \Gamma$}
\RightLabel{$\ioa^{\dag}$}
\UnaryInfC{$\R,\Gamma$}
\DisplayProof

&

\AxiomC{$\R, \R_{i}wv, w: [i] \phi, v: \phi, \Gamma$}
\RightLabel{$\stitr^{\dag}$}
\UnaryInfC{$\R, w: [i] \phi, \Gamma$}
\DisplayProof

\end{tabular}
\end{center}

\begin{center}
\begin{tabular}{c @{\hskip .5em} c}

\AxiomC{$\R, w: \agdia \phi, u:\phi, \Gamma$}
\RightLabel{$\propag^{\dag \dag}$}
\UnaryInfC{$\R, w: \agdia \phi, \Gamma$}
\DisplayProof

&

\AxiomC{$\Big\{ \R, \R_{i}w_{k}w_{j}, \Gamma \ \Big| \ 0 \leq k \leq n-1 \text{, } k+1 \leq j \leq n \Big\}$}
\RightLabel{$\choicerule$}
\UnaryInfC{$\R,\Gamma$}
\DisplayProof

\end{tabular}
\end{center}

\hrule
\caption{The labelled calculus $\calc$. The superscript $\dag$ on the $\settr$, $\stitr$, and $\ioa$ rules indicate that $v$ is an eigenvariable. The $\dag \dag$ side condition is the same as in Def. \ref{def:propagation-rules}. Last, we have $\stitr$, $\propag$, and $\choicerule$ rules for each $i \in Ag$.}
\end{figure}

The proof theoretic properties of $\ldml$ are preserved when extended with the set of propagation rules $\mathsf{PR}$ (Lem.~\ref{lm:calculus-properties}). Moreover, the nature of our propagation rules allows us to prove the admissibility of the structural rules $\stitrefl$ and $\stiteucl$, for each $i \in Ag$ (resp. Lem.~\ref{lm:refl-elimination} and \ref{lm:eucl-elimination}), which results  in the refined calculi $\calc$ (shown in Fig. 3). The proofs of Lem.~\ref{lm:calculus-properties} and \ref{lm:refl-elimination} are App.~\ref{appendix} (the latter is similar to the proof of Lem.~\ref{lm:eucl-elimination} presented here). 


\begin{lemma}\label{lm:calculus-properties} The $\ldml {+} \mathsf{PR}$ calculi have the following properties: (i) all sequents $\Lambda$ of the form $\Lambda = \R, w {:}\ \phi, w{:}\ \overline{\phi}, \Gamma$ are derivable; (ii) variable-substitution is height-preserving admissible; (iii) all inference rules are height-preserving invertible; (iv) the $\wk$, $\ctr_{\mathsf{R}}$ and $\ctr_{\mathsf{F}}$ rules are height-preserving admissible.

\end{lemma}

\begin{lemma}[$\stitrefl$-Elimination]\label{lm:refl-elimination} Every sequent $\Lambda$ derivable in $\ldml + \mathsf{PR}$ is derivable without the use of $\stitrefl$.
\end{lemma}



\begin{lemma}[$\stiteucl$-Elimination]\label{lm:eucl-elimination} Every sequent $\Lambda$ derivable in $\ldml + \mathsf{PR}$ is derivable without the use of $\stiteucl$.

\end{lemma}

\begin{proof} The result is proven by induction on the height of the given derivation. We show that the topmost instance of a $\stiteucl$ rule can be permuted upward in a derivation until it is eliminated entirely; by successively eliminating each $\stiteucl$ inference from the derivation, we obtain a derivation free of such inferences. Also, we evoke Lem.~\ref{lm:refl-elimination} and assume that all instances of $\stitrefl$ have been eliminated from the given derivation. 

\textit{Base Case.} An application of $\stiteucl$ on an initial sequent (below left) can be re-written as an instance of the $\id$ rule (below right).

\begin{center}
\begin{tabular}{c @{\hskip 1 em} c}
\AxiomC{$\R, \R_{i}wu,\R_{i}wv,\R_{i}uv, z:p, z:\overline{p}, \Gamma$}
\RightLabel{$\stiteucl$}
\UnaryInfC{$\R,\R_{i}wu,\R_{i}wv, z:p, z:\overline{p}, \Gamma$}
\DisplayProof

&

\AxiomC{}
\RightLabel{$\id$}
\UnaryInfC{$\R,\R_{i}wu,\R_{i}wv, z:p, z:\overline{p}, \Gamma$}
\DisplayProof
\end{tabular}
\end{center}
\medskip

\textit{Inductive step.} We show the inductive step for the non-trivial cases: $\stitdiar$ and $\propag$ (case (i) and (ii), respectively). All other cases are resolved by applying IH to the premise followed by an application of the corresponding rule.

\textit{(i).} Let 
 $\R_{i}uv$ be active in the $\stitdiar$ inference of the initial derivation (below (1)). Observe
  that when we apply the $\stiteucl$ rule first (below (2)), 
 the 
  atom $\R_{i}uv$ is no longer present in $\Lambda = \R,\R_{i}wu,\R_{i}wv, u : \agdia \phi, v : \phi, \Gamma$, and so, the $\stitdiar$ rule is not necessarily applicable. Nevertheless, we may apply the $\propag$ rule to derive the desired conclusion since $\agdia \agdia \in \mathcal{P}_{\Lambda}(u,v) \cap L_{i}$. Namely, 
 the fact that $\agdia \agdia \in \mathcal{P}_{\Lambda}(u,v)$ only relies on the presence of $\R_{i}wu,\R_{i}wv$ in $\Lambda$.

\begin{equation}
\begin{tabular}{c c}
\AxiomC{$\R,\R_{i}wu,\R_{i}wv,\R_{i}uv, u : \agdia \phi, v : \phi, \Gamma$}
\RightLabel{$(\langle i \rangle)$}
\UnaryInfC{$\R,\R_{i}wu,\R_{i}wv,\R_{i}uv, u : \agdia \phi, \Gamma$}
\RightLabel{$\stiteucl$}
\UnaryInfC{$\R,\R_{i}wu,\R_{i}wv, u : \agdia \phi, \Gamma$}
\DisplayProof
\end{tabular}
\end{equation}
\vspace{5pt}
\begin{equation}
\begin{tabular}{c c}
\AxiomC{$\R,\R_{i}wu,\R_{i}wv,\R_{i}uv, u : \agdia \phi, v : \phi, \Gamma$}
\RightLabel{$\stiteucl$}
\UnaryInfC{$\R,\R_{i}wu,\R_{i}wv, u : \agdia \phi, v : \phi, \Gamma$}
\RightLabel{$\propag$}
\UnaryInfC{$\R,\R_{i}wu,\R_{i}wv, u : \agdia \phi, \Gamma$}
\DisplayProof
\end{tabular}
\end{equation}
\vspace{2pt}

\textit{(ii).} Let $\Lambda_{1}$ be the first premise $\R,\R_{i}wu,\R_{i}wv,\R_{i}uv,  x:\agdia \phi, y:\phi, \Gamma$ of the initial derivation (below (3)). In the $\propag$ inference of the top derivation, we assume that $\R_{i}uv$ is active, that is, the side condition of $\propag$ is satisfied because some string $\agdia^{n} \in \mathcal{P}_{\Lambda_{1}}(x,y) \cap L_{i}$ with $n \in \mathbb{N}$. (NB. For the non-trivial case, we assume that $\agdia^{n} \in \mathcal{P}_{\Lambda_{1}}(x,y)$ relies on the presence of $\R_{i}uv \in \Lambda_{1}$, that is, the automaton $\mathcal{P}_{\Lambda_{1}}(x,y)$ makes use of transitions $u \overset{\agdia}{\longrightarrow} v$ or $v \overset{\agdia}{\longrightarrow} u$ defined relative to $\R_{i}uv$.) When we apply the $\stiteucl$ rule first in our derivation (below (4)), we can no longer rely on the relational atom $\R_{i}uv$ to apply the $\propag$ rule. However, due to the presence of $\R_{i}wu,\R_{i}wv$ in $\Lambda_{2} = \R,\R_{i}wu,\R_{i}wv, x:\agdia \phi, y:\phi, \Gamma$ we may still apply the $\propag$ rule. Namely, since $\agdia^{n} \in \mathcal{P}_{\Lambda_{1}}(x,y)$, we know there is a sequence of $n$ transitions $x \overset{\agdia}{\longrightarrow} z_{1} \overset{\agdia}{\longrightarrow} \cdots z_{n-1} \overset{\agdia}{\longrightarrow} y$ from $x$ to $y$. We replace each occurrence of $u \overset{\agdia}{\longrightarrow} v$ with $u \overset{\agdia}{\longrightarrow} w \overset{\agdia}{\longrightarrow} v$ and each occurrence of $v \overset{\agdia}{\longrightarrow} u$ with $v \overset{\agdia}{\longrightarrow} w \overset{\agdia}{\longrightarrow} u$. There will thus be a string in $\mathcal{P}_{\Lambda_{2}}(x,y) \cap L_{i}$, and so, the $\propag$ rule may be applied.

\begin{equation}
\begin{tabular}{c}
\AxiomC{$\R,\R_{i}wu,\R_{i}wv,\R_{i}uv,  x:\agdia \phi, y:\phi, \Gamma$}
\RightLabel{$\propag$}
\UnaryInfC{$\R,\R_{i}wu,\R_{i}wv,\R_{i}uv,  x:\agdia \phi, \Gamma$}
\RightLabel{$\stiteucl$}
\UnaryInfC{$\R,\R_{i}wu,\R_{i}wv,  x:\agdia \phi, \Gamma$}
\DisplayProof
\end{tabular}
\end{equation}
\vspace{5pt}
\begin{equation}
\begin{tabular}{c}
\AxiomC{$\R,\R_{i}wu,\R_{i}wv,\R_{i}uv, x:\agdia \phi, y:\phi, \Gamma$}
\RightLabel{$\stiteucl$}
\UnaryInfC{$\R,\R_{i}wu,\R_{i}wv, x:\agdia \phi, y:\phi, \Gamma$}
\RightLabel{$\propag$}
\UnaryInfC{$\R,\R_{i}wu,\R_{i}wv,  x:\agdia \phi, \Gamma$}
\DisplayProof
\end{tabular}
\end{equation}
\qed
\end{proof}


\begin{theorem}[Cut-free Completeness of $\calc$]\label{thm:cut-free-completeness-ldml}
For any formula $\phi \in \mathcal{L}^m$, if $\Vdash \phi$, then $x:\phi$ is cut-free derivable in $\calc$.
\end{theorem}

\begin{proof} Follows from Thm.~\ref{thm:g3ldm-properties}, Lem.'s \ref{lm:calculus-properties}--\ref{lm:eucl-elimination}, and the fact that, for each $i\in Ag$, the $\stitdiar$ rule is admissible, that is,  the $\stitdiar$ rule is an instance of the rule $\propag$.
\qed
\end{proof}

Last, we must 
 ensure that $\calc$ is sound. To prove this, we 
 need to stipulate how to interpret sequents on $\ldmn$-models. 
  Our definition is based on \cite{BerLyo19}: 

\begin{definition}[Interpretation, Satisfaction, Validity]\label{def:sequent-satisfaction-validity} Let $M$ be an $\ldmn$-model with domain $W$
, $\Lambda = \R,\Gamma$ a labelled sequent, and Lab the set of labels. Let $I$ be an \emph{interpretation function} 
 mapping labels to worlds: \textit{i.e.} $I{:} \ Lab \mapsto W$.

$\Lambda$ is \emph{satisfied} in $M$ with $I$ iff for all relational atoms $\R_{i}xy \in \R$, if $\R_{i} x^{I}y^{I}$ holds in $M$, then there exists some $z : \phi \in \Gamma$ such that $M, z^{I} \Vdash \phi$.

$\Lambda$ is \emph{valid} iff it is satisfiable in every $M$ with any interpretation function $I$.

\end{definition}

\begin{theorem}[$\calc$ Soundness]\label{thm:soundness-ldmnl}
Every sequent derivable in $\calc$ is valid.
\end{theorem}

\begin{proof} We know by Thm.~\ref{thm:g3ldm-properties} that all rules of $\calc$, with the exception of $\propag$, preserve validity. Details of the $\propag$ rule are given in App.~\ref{appendix}.
\qed
\end{proof}

\section{Proof-Search and Decidability}\label{sec_proofsearch}

In this section, we provide a class of proof-search algorithms, each deciding a logic $\ldmsa$ (with $n \in \mathbb{N}$). (
We 
 use $1$ to denote the agent in the single-agent setting.) In the single-agent case, the independence of agents condition is trivially satisfied, meaning we can omit the $\ioa$ rule from each calculus and from consideration during proof-search. We end the section by commenting on the more complicated 
 multi-agent setting.


In what follows, we prove that derivations in $\ldmlsa$ need only use \emph{forestlike sequents}. The forestlike structure of a sequent $\Lambda$ refers to a graph corresponding to the sequent. This control in sequential structure is what allows us to adapt methods from~\cite{TiuIanGor12} to $\ldmlsa$, and produce a proof-search algorithm that decides $\ldmsa$, for each $n \in \mathbb{N}$. Let us start by making the aforementioned notions precise.

\begin{definition}[Sequent Graph]\label{def:graph-of-sequent} We define a \emph{graph} $G$ to be a tuple $(V,E,L)$, where $V$ is the non-empty \emph{set of vertices}, the \emph{set of edges} $E \subseteq V \times V$, and $L$ is the \emph{labelling function} that maps edges from $E$ into some non-empty set $S$ and vertices from $V$ into some non-empty set $S'$.

Let $\Lambda = \R,\Gamma$ be a labelled sequent and let $Lab(\Lambda)$ be the set of labels in $\Lambda$. The \emph{graph of $\Lambda$}, denoted $G(\Lambda)$, is the tuple $(V,E,L)$, where (i) $V = Lab(\Lambda)$, (ii) $(w,u) \in E$ and $L(w,u) = i$ iff $\R_{i}wu \in \R$, and (iii) $L(w) = \phi$ iff $w : \phi \in \Gamma$.

\end{definition}

\begin{example}\label{eg2:graph-prop-automaton} The sequent \emph{graph} $G(\Lambda)$ 
 corresponding to the labelled sequent $\Lambda$ = $\R_{1}xy,\R_{1}zx,x:p,y:\overline{p} \vee q, z:r,z:\Diamond q$ is shown below:

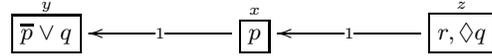
\begin{small}
\begin{center}
\begin{tabular}{c c c 
}
\xymatrix{
		 \overset{y}{\boxed{\overline{p} \vee q}} &  & \overset{x}{\boxed{p}}\ar[ll]|-{1}  &	& \overset{z}{\boxed{r,\Diamond q}}\ar[ll]|-{1} \\
}
%
\end{tabular}
\end{center}
\end{small}
\end{example}



\begin{definition}[Tree, Forest, Forestlike Sequent, Choice-tree]\label{def:forestlike-sequent} We say that a graph $G = (V,E,L)$ is a \emph{tree} iff there exists a node $w$, called the \emph{root}, such that there is exactly one directed path from $w$ to any other node $u$ in the graph. We say that a graph is a \emph{forest} iff it consists of a disjoint union of trees.

A sequent $\Lambda$ is \emph{forestlike} iff its graph $G(\Lambda)$ is a forest. We refer to each disjoint tree in the graph of a forestlike sequent as a \emph{choice-tree} and for any label $w$ in $\Lambda$, we let $CT(w)$ represent the choice-tree that $w$ belongs to.

\end{definition}

The above 
notions 
will be 
significant for our proof-search algorithms, for example: 

\begin{remark}\label{rmk:forest-structure-remark}
 When interpreting a sequent, each choice-tree that occurs in the graph of the sequent is a syntactic representation of an equivalence class of $\R_{1}$ (i.e., a choice-cell for agent $1$). Using this insight, 
 we know that if agent $1$ is restricted to $n$-many choices, then if there are $m > n$ choice-trees in the sequent, at least two 
   choice-trees must correspond to the same equivalence class in $\R_{1}$. We use this observation to specify how 
   $\mathsf{APC}_{n}^{1}$ is applied in the algorithm.

\end{remark}

The following definitions introduce the necessary tools for the algorithms:

\begin{definition}[Saturation, $\Box-, \agboxone$-realization, $\Diamond-$, $\agdiaone$-propagated]\label{def:saturation-realization-propagation} Let $\Lambda$ be a forestlike sequent 
 and let w be a label in $\Lambda$.
 
The label $w$ is \emph{saturated} iff the following hold: (i) If $w : \phi \in \Lambda$, then $w : \overline{\phi} \not\in \Lambda$, (ii) if $w:\phi \lor \psi \in \Lambda$, then $w:\phi \in \Lambda$ and $w:\psi \in \Lambda$, (iii) if $w:\phi \land \psi \in \Lambda$, then $w:\phi \in \Lambda$ or $w:\psi \in \Lambda$.

A label $w$ in $\Lambda$ is \emph{$\Box$-realized} iff for every $w : \Box \phi \in \Lambda$, there exists a label $u$ such that $u:\phi \in \Lambda$. A label $w$ in $\Lambda$ is \emph{$\agboxone$-realized} iff for every $w : [1] \phi \in \Lambda$, there exists a label $u$ in $CT(w)$ such that $u:\phi \in \Lambda$.

A label $w$ in $\Lambda$ is \emph{$\Diamond$-propagated} iff for every $w : \Diamond \phi \in \Lambda$, we have 
$u:\phi \in \Lambda$ for all labels $u$ in $\Lambda$. A label $w$ in $\Lambda$ is \emph{$\agdiaone$-propagated} iff for every $w : \agdiaone \phi \in \Lambda$, we have 
 $u:\phi \in \Lambda$ for all labels $u$ in $CT(w)$.

\end{definition}

\begin{definition}[$n$-choice Consistency]\label{def:n-Choice-consistency} Let $\Lambda$ be a forestlike sequent and let our logic be $\ldmsa$ with $n > 0$. We say that $\Lambda$ is \emph{$n$-choice consistent} iff $G(\Lambda)$ contains at most $n$-many choice-trees.

\end{definition}


\begin{definition}[Stability]\label{def:stability} A forestlike labelled sequent $\Lambda$ is \emph{stable} iff (i) all labels $w$ in $\Lambda$ are saturated, (ii) all labels are $\Box$- and $[1]$-realized, (iii) all labels are $\Diamond$- and $\agdiaone$-propagated,  and (iv) $\Lambda$ is $n$-choice consistent.
\end{definition}

We are now able to 
define our proof-search algorithms for the logics $\ldmsa$. The algorithms are provided in Fig.~4 
 and are inspired by~\cite{TiuIanGor12}. We emphasize that the execution of instruction 4 in Fig.~4 
  corresponds to an instance of the $(\mathsf{Pr}_{1})$ rule. The algorithms are correct (Thm. \ref{thm:correctness-prove}) and terminate (Thm. \ref{thm:termination-prove}). 
Last, 
 Lem. \ref{lm:forestlike-invariance} ensures that the concepts of realization, propagation, $n$-choice consistency, and stability are defined at each stage of the computation (Def. \ref{def:saturation-realization-propagation} - \ref{def:stability}). The proofs of Lem. \ref{lm:forestlike-invariance} and Thm. \ref{thm:termination-prove} can be found in App.~\ref{appendix}. 

\begin{figure}[t] 
\label{fig:proof-search-algorithm}
\noindent\hrule
\texttt{
\begin{center}
Function Prove\textsubscript{n}(Sequent $\R,\Upgamma$) : Boolean \\
\end{center}
\begin{flushleft}
1. If $\R, \Upgamma = \R, w:p, w:\overline{p}, \Upgamma'$, return true.\\
2. If $\R, \Upgamma$ is stable, return false.\\
3. If some label $w$ in $\R, \Upgamma$ is not saturated, then:\\
\ \ \ \ \hangindent=3em (i) If $w : \phi \lor \psi \in \R, \Upgamma$, but either $w : \phi \not\in \R, \Upgamma$ or $w : \psi \not\in \R, \Upgamma$, then let $\R,\Upgamma' = \R, w:\phi, w:\psi, \Upgamma$ and return Prove\textsubscript{n}($\R, \Upgamma'$).\\
\ \ \ \ \hangindent=3em (ii) If $w : \phi \land \psi \in \R, \Upgamma$, but $w : \phi, w : \psi \not\in \R, \Upgamma$, then let $\R,\Upgamma_{1} = \R, w:\phi, \Upgamma$, let $\R,\Upgamma_{2} = \R, w:\psi, \Upgamma$, and return false if Prove\textsubscript{n}($\R, \Upgamma_{i}$) = false for some $i \in \{1,2\}$, and return true otherwise.\\
4. If some label $w$ in $\R, \Upgamma$ is not $\agdiaone$-propagated, then there is a label $u$ in $CT(w)$ such that $u : \phi \not\in \Upgamma$. Let $\R,\Upgamma' = \R, u: \phi, \Upgamma$ and return Prove\textsubscript{n}($\R, \Upgamma'$).\\
5. If some label $w$ in $\R, \Upgamma$ is not $\Diamond$-propagated, then there is a label $u$ such that $u : \phi \not\in \Upgamma$. Let $\R,\Upgamma' = \R, u: \phi, \Upgamma$ and return Prove\textsubscript{n}($\R, \Upgamma'$).\\
6. If there is a label $w$ that is not $[1]$-realized, then there is a $w : [1] \phi \in \Upgamma$ such that $u : \phi \not\in \Upgamma$ for every label $u \in CT(w)$. Let $\R', \Upgamma' = \R, \R_{1}wv, v : \phi, \Upgamma$ with $v$ fresh and return Prove\textsubscript{n}($\R', \Upgamma'$).\\
7. If there is a label $w$ that is not $\Box$-realized, then there is a $w : \Box \phi \in \Upgamma$ such that $u : \phi \not\in \Upgamma$ for every label $u$ in $\R,\Upgamma$. Let $\R, \Upgamma' = \R, v : \phi, \Upgamma$ with $v$ fresh and return Prove\textsubscript{n}($\R, \Upgamma'$).\\
8. If $\R, \Upgamma$ is not $n$-choice consistent, then let $\R_{k,j},\Upgamma = \R, \R_{1}w_{k}w_{j}, \Upgamma$ (with $0 \leq k \leq n-1$ and $k+1 \leq j \leq n$) and where each $w_{k}$ and $w_{j}$ are distinct roots of choice-trees in $\R, \Upgamma$. Return false if Prove\textsubscript{n}($\R_{k,j}, \Upgamma$) = false for some $k$ and $j$, and return true otherwise.
\end{flushleft}
}
\vspace{-5pt}
\hrule
\caption{The proof-search algorithms 
 for $\ldmsa$ with $n > 0$. The algorithm for $\mathsf{Ldm}_{0}^{1}$ is obtained by deleting line 8.}
 \vspace{-5pt}
\end{figure}

\begin{lemma}\label{lm:forestlike-invariance} Every labelled sequent generated throughout the course of computing $\mathtt{Prove_{n}}(w:\phi)$ is forestlike.
\end{lemma}

\begin{theorem}[Correctness]\label{thm:correctness-prove} \emph{(i)} 
 If $\mathtt{Prove_{n}(w:\phi)}$ returns $\mathtt{true}$, then $w:\phi$ is $\ldmlsa$-provable. \emph{(ii)}
  If $\mathtt{Prove_{n}(w:\phi)}$ returns $\mathtt{false}$, then $w:\phi$ is not $\ldmlsa$-provable. 
\end{theorem}

\begin{proof} (i) 
It suffices to observe that each step of $\mathtt{Prove_{n}(\cdot)}$ is a backwards application of a rule in $\ldmlsa$, and so, if the proof-search algorithm returns $\mathtt{true}$, 
 the formula $w : \phi$ is derivable in $\ldmlsa$ with arbitrary label $w$. 
   
(ii) To prove 
this 
statement, 
 we assume that $\mathtt{Prove_{n}(w:\phi)}$ returned $\mathtt{false}$ and show that we can construct a counter-model for 
    $\phi$.
    By assumption, 
    we know that a stable sequent $\Lambda$ was generated with $w : \phi \in \Lambda$. We define our counter-model $M =(W,\R_{1},V)$ as follows: 
    $W = Lab(\Lambda)$; 
     $\R_{1}uv$ 
      \textit{iff} $\mathcal{P}_{\Lambda}(u,v) \cap L_{1} \neq \emptyset$; and 
      $w \in V(p)$ \textit{iff} $w : \overline{p} \in \Lambda$. 

We argue that $F = (W, \R_{1})$ is an $\ldmn$-frame. It is easy to see that $W\neq \emptyset$ 
   (at the very least, the label $w$ must occur in $\Lambda$). Moreover, condition (C2) is trivially satisfied in the single-agent setting. 
   We prove 
  (C1) and (C3):

(C1) 
 We need to prove that $\R_{1}$ is 
 (i) reflexive and (ii) Euclidean. To prove (i)
 , it suffices 
  to show that for each $u\in Lab(\Lambda)$ there exists a string $\omega$ in both $\mathcal{P}_{\Lambda}(u,u)$ and $L_{1}$. By Def. \ref{def:propagation-automaton}, we know that $\varepsilon \in \mathcal{P}_{\Lambda}(u,u)$ since $u$ is both the initial and accepting state. Also, by Def. \ref{def:appl_lang} we know that $\varepsilon \in L_{1}$. 
  To prove (ii), 
  we assume that $\R_{1}wu$ and $\R_{1}wv$ hold, and 
  show that $\R_{1}uv$ holds as well. By our assumption, there exist strings $\agdiaone^{k} \in \mathcal{P}_{\Lambda}(w,u) \cap L_{1}$ and $\agdiaone^{m} \in \mathcal{P}_{\Lambda}(w,v) \cap L_{1}$ (with $k,m \in \mathbb{N}$). It is not difficult to prove that if $\agdiaone^{k} \in \mathcal{P}_{\Lambda}(w,u)$, then $\agdiaone^{k} \in \mathcal{P}_{\Lambda}(u,w)$, and also that if $\agdiaone^{k} \in \mathcal{P}_{\Lambda}(u,w)$ and $\agdiaone^{m} \in \mathcal{P}_{\Lambda}(w,v)$, then $\agdiaone^{k+m} \in \mathcal{P}_{\Lambda}(u,v)$. Hence, we know $\agdiaone^{k+m} \in \mathcal{P}_{\Lambda}(u,v)$, which, together with $\agdiaone^{k+m} \in L_{1}$ (Def. \ref{def:appl_lang}), gives us the desired $\R_{1}uv$.

(C3) By assumption we know $\Lambda$ is stable. Consequently, 
  when $n > 0$ for $\ldmlsa$, 
 the sequent $\Lambda$ must be $n$-choice consistent. 
 Hence, the graph of $\Lambda$ must contain $k \leq n$ choice-trees. Condition (C3) follows straightforwardly. 

Since $F$ is an $\ldmn$-frame, $M$ is an $\ldmn$-model. We 
 show by induction on the complexity of $\psi$ that for any $u : \psi \in \Lambda$, $M, u \not\Vdash \psi$. Consequently, 
 $M$ is a counter-model for $\phi$, and so, 
 by Thm. \ref{thm:soundness-ldmnl}, 
 we know $w : \phi$ is not provable in $\ldmlsa$.
 

\textit{Base Case.} Assume 
 $u : p \in \Lambda$. Since $\Lambda$ is stable, we know that $ u : \overline{p} \not\in \Lambda$. 
Hence, 
 by the definition of $V$, we know that $u \not\in V(p)$, implying that $M, u \not \Vdash p$. 

\textit{Inductive Step.} We consider each connective 
 in turn. 
  (i) Assume that $u : \theta \lor \chi \in \Lambda$. Since $\Lambda$ is stable, it is saturated, meaning that $u : \theta, u: \chi \in \Lambda$. Hence, by IH, $M,u \not\Vdash \theta$ and $M,u \not\Vdash \chi$, which implies that $M,u \not\Vdash \theta \lor \chi$. 
  (ii) The case $u: \theta\land\chi\in \Lambda$ is similar to the previous case. 
   (iii) 
   Assume 
    $u : \agdiaone \theta \in \Lambda$. Since 
    $\Lambda$ is stable, we know that every label is $\agdiaone$-propagated. Therefore, for all labels $v \in CT(u)$ 
    we have $v : \theta \in \Lambda$. By IH, $M,v \not\Vdash \theta$ for all $v \in CT(u)$. In general, the definition of $\R_{1}$ implies that $\R_{1}xy$ iff $y \in CT(x)$. The former two statements imply that $M,v \not\Vdash \theta$ for all $v$ such that $\R_{1}uv$, and so, $M,u \not\Vdash \agdiaone \theta$. 
    (iv) 
    Assume that $u : \Diamond \theta \in \Lambda$. Since $\Lambda$ is stable, every label is $\Diamond$-propagated, which implies that for all labels $v$ in $\Lambda$, $v : \theta \in \Lambda$. By IH, this implies that for all $v \in W$
   , $M,v \not\Vdash \theta$. Thus, $M,u \not\Vdash \Diamond \theta$. 
   (v) 
   Assume 
   $u:[1] \theta \in \Lambda$. Since $\Lambda$ is stable, we know 
   every label in $\Lambda$ is $[1]$-realized. 
    Therefore, there exists a label $v$ in $CT(u)$ such that $v : \theta \in \Lambda$. By IH, we conclude that $M,v \not\Vdash \theta$. Moreover, since $\R_{1}xy$ iff $y \in CT(x)$, we also 
     know that $\R_{1}uv$, which implies 
     $M,u \not\Vdash [1] \psi$. 
     (vi) 
     Assume 
     $u : \Box \theta \in \Lambda$. Since $\Lambda$ is stable, we know that every label is $\Box$-realized. Consequently, there 
     exists a label $v$ such that $v : \theta \in \Lambda$. By IH, we 
      conclude 
       $M,v \not\Vdash \theta$; hence, $M,u \not\Vdash \Box \theta$.
\qed
\end{proof}

\begin{theorem}[Termination]\label{thm:termination-prove} For each 
 formula $w\ {:}\ \phi$, $\mathtt{Prove_{n}(w{:}\ \phi)}$ terminates.
\end{theorem}

\begin{corollary}[Decidability and FMP]
For each $n \in \mathbb{N}$, the logic $\ldmsa$ is decidable and has the finite model property.
\end{corollary}

\begin{proof} Follows from Thm. \ref{thm:correctness-prove} and \ref{thm:termination-prove} above. The finite model property follows from the fact that the counter-models constructed in Thm. \ref{thm:correctness-prove} are all finite.
\qed
\end{proof}

Additionally, 
 from a computational viewpoint, 
 it is interesting to know if completeness is preserved under a restricted class of 
 sequents (cf.~\cite{CiaLyoRam18}). Indeed, Lem. \ref{lm:forestlike-invariance}, Thm. \ref{thm:correctness-prove} and Thm. \ref{thm:termination-prove}, imply that completeness is preserved when we restrict $\ldmlsa$ derivations to forestlike sequents
 ; that is, when inputting a formula into our algorithms, the sequent produced at each step of the computation will be forestlike. 
 Interestingly, this result was obtained 
 via our proof-search algorithms. 

\begin{corollary}[Forestlike Derivations]\label{cor:forest-like-derivations}
For each $n \in \mathbb{N}$, if a labelled formula $w : \phi$ is derivable in $\ldmlsa$, then it is derivable using only forestlike sequents.
\end{corollary}


\subsubsection{A Note on the Multi-Agent Setting of $\calc$.}\label{sec-multiagent}

As a concluding remark, we briefly touch upon extending the current results to the multi-agent calculi $\calc$. 
 In the multi-agent setting (when $n = 0$), 
 our sequents have the structure of \textit{directed acyclic graphs} (i.e., directed graphs free of cycles), due to the independence of agents rule $\ioa$. In such graphs, one can easily recognize loop-nodes---i.e., 
 a path from an ancestor node to the alleged loop-node such that both nodes are labelled with the same multiset of formulae---
 and use this information to bound the depth of the sequent during proof-search (cf.~\cite{TiuIanGor12}). 

The main challenge concerns the 
$\ioa$ rule, which 
 when applied bottom-up during proof-search, introduces a fresh label $v$ to the sequent. As a consequence, one must ensure that if proof-search terminates in a counter-model construction, this label $v$ satisfies the independence of agents condition in that model. At first glance, one might conjecture that for every application of the $\ioa$ rule an additional application of the rule is needed to saturate the independence of agents condition. Of course, in such a case the algorithm will not terminate with a sequent that is readily convertible to a counter-model. Fortunately, it turns out that only finitely many applications of $\ioa$ are needed to construct a counter-model satisfying independence of agents. The authors have planned to devote their future work to answer this open problem for the multi-agent setting.

\section{Conclusion}\label{sec_concl}

This paper introduced the first cut-free complete calculi for the class of multi-agent $\ldmn$ logics, 
 introduced in \cite{Xu94-2}.
 We adapted propagation rules, discussed in~\cite{TiuIanGor12}, in order to refine the multi-agent $\ldml$ labelled calculi and generate the proof-search friendly $\calc$ calculi. For the single agent case, we provided a class of terminating proof-search algorithms, each deciding a logic $\ldmsa$ (with $n \in \mathbb{N}$), including counter-model extraction from failed proof-search.

As discussed in Sec.~\ref{sec-multiagent}, we plan to devote future research to leveraging the current results for the multi-agent setting and to provide terminating proof-search procedures for the entire $\ldmn$ class. As a natural extension, we aim to implement the proof-search algorithms from Sec.~\ref{sec_proofsearch} in \textsc{Prolog} (e.g., as in \cite{GirLelOliPozVit17}). Additionally, we plan to expand the current framework to include deontic $\stit$ operators (e.g., from \cite{Hor01,Mur04}) with the goal of automating normative, agent-based reasoning. Last, it is shown in~\cite{BalHerTro08} that $\ldmp{0}{1}$ has an NP-complete satisfiability problem and each logic $\ldmp{0}{m}$, with $m > 0$, is NEXPTIME-complete. 
Along with expanding our proof-search algorithms to the class of 
 all $\ldmn$ logics, we aim to investigate the complexity 
 and 
  optimality of our associated
   algorithms.

 \bibliographystyle{splncs04}
%



\appendix

\section*{Appendix}
\section{Proofs}\label{appendix}

\begin{customthm}{\ref{thm:g3ldm-properties}} The $\ldml$ calculi have the following properties:

\begin{enumerate}

\item  All sequents of the form $\R, w : \phi, w: \overline{\phi}, \Gamma$ are derivable;

\item Variable-substitution is height-preserving admissible;

\item All inference rules are height-preserving invertible;

\item The weakening rule $\wk$ and two contraction rules $\ctr$ (below) are height-preserving admissible;

\begin{center}
\begin{tabular}{c @{\hskip 1em} c @{\hskip 1em} c}
\AxiomC{$\R,\Gamma$}
\RightLabel{$\wk$}
\UnaryInfC{$\R,\R',\Gamma',\Gamma$}
\DisplayProof

&

\AxiomC{$\R,\R',\R',\Gamma$}
\RightLabel{$\ctr_{\mathsf{R}}$}
\UnaryInfC{$\R,\R',\Gamma$}
\DisplayProof

&

\AxiomC{$\R,\Gamma',\Gamma',\Gamma$}
\RightLabel{$\ctr_{\mathsf{F}}$}
\UnaryInfC{$\R,\Gamma',\Gamma$}
\DisplayProof
\end{tabular}
\end{center}

\item The cut rule $\cut$ (below) is admissible;

\begin{center}
\begin{tabular}{c}
\AxiomC{$\R,x:\phi,\Gamma$}
\AxiomC{$\R,x:\overline{\phi},\Gamma$}
\RightLabel{$\cut$}
\BinaryInfC{$\R,\Gamma$}
\DisplayProof
\end{tabular}
\end{center}

\item For every formula $\phi \in \mathcal{L}$, $\vdash w: \phi$ is derivable in $\ldml$ if and only if $\vdash_{\ldm} \phi$, i.e. $\ldml$ is sound and complete relative to $ \ldm$.

\end{enumerate}
\end{customthm}

\begin{proof} Clause (1) is proven by induction on the complexity of the formula $\phi$, whereas clauses (2)-(4) are proven by induction on the height of the given derivation. Moreover, clause (5) is proven by induction on the complexity of the cut formula $\phi$ with a subinduction on the sum of the heights of the premises of the $\cut$ rule. All proofs are similar to proofs of the same properties for modal calculi presented in~\cite{Neg05}.

Concerning clause (5), the proof of soundness is argued along the same lines as in~\cite[Sec.~5.1]{Neg09} and uses Def.~\ref{def:sequent-satisfaction-validity}. Completeness is proven by showing that all axioms of $\ldmn$ can be derived and that if the premise of an inference rule is derivable, then so is the conclusion. All proofs are fairly simple, with the exception of the derivation of the $\ioa$ and $\apc$ axioms, which we provide below.\\


\noindent
$\ioa$:

\begin{center}
\begin{small}
\resizebox{\columnwidth}{!}{
\begin{tabular}{c}

\AxiomC{$\Pi_{1}$}
\AxiomC{$\cdots$}
\AxiomC{$\Pi_{m}$}

\TrinaryInfC{$\R_{1}y_{1}v, \ldots, \R_{m}y_{m}v,x : \Box \lb 1 \rb \overline{\phi}_{1}, \ldots, x: \Box \lb m \rb \overline{\phi}_{m}, y_{1} : \lb 1 \rb \overline{\phi}_{1}, \ldots, y_{m}: \lb m \rb \overline{\phi}_{m}, x: \Diamond ([1] \phi_{1} \wedge \cdots \wedge [m] \phi_{m}), v: [1] \phi_{1} \wedge \cdots \wedge [m] \phi_{m}$}

\UnaryInfC{$\R_{1}y_{1}v, \ldots, \R_{m}y_{m}v, x : \Box \lb 1 \rb \overline{\phi}_{1}, \ldots, x: \Box \lb m \rb \overline{\phi}_{m},y_{1} : \lb 1 \rb \overline{\phi}_{1}, \ldots, y_{m}: \lb m \rb \overline{\phi}_{m}, x: \Diamond ([1] \phi_{1} \wedge \cdots \wedge [m] \phi_{m})$}
\UnaryInfC{$x : \Box \lb 1 \rb \overline{\phi}_{1}, \ldots, x: \Box \lb m \rb \overline{\phi}_{m},y_{1} : \lb 1 \rb \overline{\phi}_{1}, \ldots, y_{m}: \lb m \rb \overline{\phi}_{m}, x: \Diamond ([1] \phi_{1} \wedge \cdots \wedge [m] \phi_{m})$}
\UnaryInfC{$x : \Box \lb 1 \rb \overline{\phi}_{1}, \ldots, x: \Box \lb m \rb \overline{\phi}_{m}, x: \Diamond ([1] \phi_{1} \wedge \cdots \wedge [m] \phi_{m})$}
\UnaryInfC{$x : \Box \lb 1 \rb \overline{\phi}_{1} \vee \cdots \vee \Box \lb m \rb \overline{\phi}_{m} \vee \Diamond ([1] \phi_{1} \wedge \cdots \wedge [m] \phi_{m})$}
\DisplayProof

\end{tabular}
}
\end{small}
\end{center}

\noindent
with $\Pi_{i}$ (for $1 \leq i \leq m$) given by:\\

\begin{center}
\begin{small}
\begin{tabular}{c}
\AxiomC{$\R_{i}vu,\R_{i}y_{i}v, \R_{i}y_{i}u, ..., y_{i} : \lb i \rb \overline{\phi}_{i}, u : \overline{\phi}_{i},  u : \phi_{i}$}
\UnaryInfC{$\R_{i}vu,\R_{i}y_{i}v, \R_{i}y_{i}y_{i}, \R_{i}vy_{i}, \R_{i}y_{i}u, ..., y_{i} : \lb i \rb \overline{\phi}_{i}, u : \phi_{i}$}
\UnaryInfC{$\R_{i}vu,\R_{i}y_{i}v, \R_{i}y_{i}y_{i}, \R_{i}vy_{i}, ..., y_{i} : \lb i \rb \overline{\phi}_{i}, u : \phi_{i}$}
\UnaryInfC{$\R_{i}vu,\R_{i}y_{i}v, \R_{i}y_{i}y_{i}, ..., y_{i} : \lb i \rb \overline{\phi}_{i}, u : \phi_{i}$}
\UnaryInfC{$\R_{i}vu,\R_{i}y_{i}v, ..., y_{i} : \lb i \rb \overline{\phi}_{i}, u : \phi_{i}$}
\UnaryInfC{$\R_{i}y_{i}v, ..., y_{i} : \lb i \rb \overline{\phi}_{i}, v: [i] \phi_{i}$}
\DisplayProof
\end{tabular}
\end{small}
\end{center}

\noindent
$\apc$:

\begin{center}
\begin{small}
\resizebox{\columnwidth}{!}{
\begin{tabular}{c}

\AxiomC{$\Big\{ \Pi_{k,j} \ \Big| \ 0 \leq k \leq n-1 \text{, } k+1 \leq j \leq n \Big\}$}


\UnaryInfC{$w_{0}:\Box\langle i \rangle \overline{\phi_1}, w_{1} : \overline{\phi_{1}},
w_{0}: \Box (\phi_1\lor\langle i\rangle \overline{\phi_2}), w_{2} : \phi_{1}, w_{2} : \agdia \overline{\phi_2}, \ldots,
w_{0}:\Box (\phi_{1} \lor \cdots \phi_{n-1} \lor \agdia \overline{\phi_{n}}),
w_{n}: \phi_{1}, \ldots, w_{n} : \phi_{n-1}, w_{n} : \agdia \overline{\phi_{n}},
 w_{0}:\phi_1, \ldots, w_{0}: \phi_n$}
\UnaryInfC{$w_{0}:\Box\langle i \rangle \overline{\phi_1}, w_{0}: \Box (\phi_1\lor\langle i\rangle \overline{\phi_2}), \ldots, w_{0}:\Box (\phi_{1} \lor \cdots \phi_{n-1} \lor \agdia \overline{\phi_{n}}), w_{0}:\phi_1, \ldots, w_{0}: \phi_n$}
\UnaryInfC{$w_{0}:\Box\langle i \rangle \overline{\phi_1} \lor \Box (\phi_1\lor\langle i\rangle \overline{\phi_2}) \lor \cdots \lor \Box (\phi_{1} \lor \cdots \phi_{n-1} \lor \agdia \overline{\phi_{n}}) \lor\phi_1\lor \cdots \lor \phi_n$}
\DisplayProof 
\end{tabular}
}
\end{small}
\end{center}

\noindent
The $\Pi_{0,j}$ (for $1 \leq j \leq n$) derivations are shown below left, and the $\Pi_{k,j}$ (for $0 < k \leq n-1$ and $k+1 \leq j \leq n$) derivations are shown below right.

\begin{center}
\begin{small}
\resizebox{\columnwidth}{!}{
\begin{tabular}{c @{\hskip 1em} c}
\AxiomC{$\R_{i}w_{0}w_{j}, \R_{i}w_{j}w_{0},, \ldots, w_{0}: \phi_{j}, w_{j}: \agdia \overline{\phi_{j}}, w_{0} : \overline{\phi_{j}}$}
\UnaryInfC{$\R_{i}w_{0}w_{j}, \R_{i}w_{j}w_{0}, \ldots, w_{0}: \phi_{j} , w_{j}: \agdia \overline{\phi_{j}}$}
\dashedLine
\UnaryInfC{$\R_{i}w_{0}w_{j}, \ldots, w_{0}: \phi_{j} , w_{j}: \agdia \overline{\phi_{j}}$}
\DisplayProof

&

\AxiomC{$\R_{i}w_{k}w_{j}, \ldots, w_{k}: \agdia \overline{\phi_{k}}, w_{j} : \overline{\phi_{k}}, w_{j}: \phi_{k}$}
\UnaryInfC{$\R_{i}w_{k}w_{j}, \ldots, w_{k}: \agdia \overline{\phi_{k}}, w_{j}: \phi_{k}$}
\DisplayProof
\end{tabular}
}
\end{small}
\end{center}

Note that the dashed lines in the above derivations represent applications of the $(\mathsf{sym}_{i})$ rule (below left); however, we may apply this rule since it is admissible in $\ldml$ (as shown below right).

\begin{center}
\begin{small}
\begin{tabular}{c @{\hskip 2em} c}
 \AxiomC{$\R, \R_{i}wu, \R_{i}uw, \Gamma$}
 \RightLabel{$(\mathsf{sym}_{i})$}
 \UnaryInfC{$\R, \R_{i}wu,\Gamma$}
 \DisplayProof
 
 &
 
\AxiomC{$\R, \R_{i}wu, \R_{i}uw, \Gamma$}
\RightLabel{Lem.~\ref{thm:g3ldm-properties}, Clause (4)}
\UnaryInfC{$\R, \R_{i}wu, \R_{i}uu, \R_{i}uw, \Gamma$}
\RightLabel{$\stiteucl$}
\UnaryInfC{$\R, \R_{i}wu, \R_{i}ww, \Gamma$}
\RightLabel{$\stitrefl$}
\UnaryInfC{$\R, \R_{i}wu,\Gamma$}
\DisplayProof
\end{tabular}
\end{small}
\end{center} 

\end{proof}

\begin{customlem}{\ref{lm:calculus-properties}} The calculus $\ldml + \mathsf{PR}$ has the following properties: (i) All sequents of the form $\R, w : \phi, w: \overline{\phi}, \Gamma$ are derivable, (ii) Variable-substitution is height-preserving admissible, (iii) All inference rules are height-preserving invertible, (iv) The $\wk$, $\ctr_{\mathsf{R}}$ and $\ctr_{\mathsf{F}}$ rules are height-preserving admissible.

\end{customlem}


\begin{proof} 

(i) is proved by induction on the complexity of $\phi$ and (ii)-(iv) are shown by induction on the height of the given derivation. All proofs are routine, so we only 
prove the most significant result: height-preserving admissibility of contraction. 

We proceed 
 by induction on the height of the given derivation. With the exception of 
 $\propag$
 , the proof is exactly the same as for Thm. \ref{thm:g3ldm-properties} clause 4 (see 
 ~\cite{Neg05,Neg09} for details). Hence, we only 
  prove the $\propag$ case in the inductive step.

First, we show that 
 $\ctr_{\mathsf{F}}$ 
  can be permuted with $\propag$. 
 There are two cases: 
 either the derivation performs a formula contraction solely in the formula-context $\Gamma$ (below left), or the derivation performs a contraction with the auxiliary formula $w:\agdia \phi$ (below right) (potentially performing a contraction in $\Gamma$): 

\begin{center}
\begin{tabular}{c @{\hskip 2em} c}
\AxiomC{$\R, w:\agdia \phi, u : \phi, \Gamma$}
\RightLabel{$\propag$}
\UnaryInfC{$\R, w:\agdia \phi, \Gamma$}
\RightLabel{$\ctr_{\mathsf{F}}$}
\UnaryInfC{$\R, w:\agdia \phi, \Gamma'$}
\DisplayProof

&

\AxiomC{$\R, w:\agdia \phi, w:\agdia \phi, u : \phi, \Gamma$}
\RightLabel{$\propag$}
\UnaryInfC{$\R, w:\agdia \phi, w:\agdia \phi, \Gamma$}
\RightLabel{$\ctr_{\mathsf{F}}$}
\UnaryInfC{$\R, w:\agdia \phi, \Gamma'$}
\DisplayProof

\end{tabular}
\end{center}

In both cases, any sequence of transitions between $w$ and $u$ will be preserved, meaning that the $\ctr_{\mathsf{F}}$ rule may be freely permuted with the $\propag$ rule without affecting the side condition of the propagation rule. 

Secondly, we consider the $\ctr_{\mathsf{R}}$ case, which has the following form:

\begin{center}
\begin{tabular}{c @{\hskip 2em} c}
\AxiomC{$\R, w:\agdia \phi, u : \phi, \Gamma$}
\RightLabel{$\propag$}
\UnaryInfC{$\R, w:\agdia \phi, \Gamma$}
\RightLabel{$\ctr_{\mathsf{R}}$}
\UnaryInfC{$\R', w:\agdia \phi, \Gamma$}
\DisplayProof
\end{tabular}
\end{center}

Since the $\propag$ rule is applied first, we know there exists a sequence of transitions $w \overset{\agdia}{\longrightarrow} v_{1} \overset{\agdia}{\longrightarrow} \cdots v_{n} \overset{\agdia}{\longrightarrow} u$ from $w$ to $u$. Notice that $\ctr_{\mathsf{R}}$ contracts identical relational atoms in $\R$, resulting in a single copy still present in $\R'$. 
 Hence, if we apply $\ctr_{\mathsf{R}}$ first, we may still apply $\propag$ afterwards, since the same sequence of transitions from $w$ to $u$ remains present 
  in $R'$. 
  See~\cite[Lem. 6.12]{GorPosTiu11} for further details on the preservation of transitions under contractions.

\end{proof}

\begin{customlem}{\ref{lm:refl-elimination}}[$\stitrefl$-Elimination] Every sequent $\Lambda$ derivable in $\ldml + \mathsf{PR}$ is derivable without the use of $\stitrefl$.
\end{customlem}

\begin{proof} The result is proven by induction on the height of the given derivation. We assume that $\refl$ is used once as the last inference in the given derivation. The general result follows by successively eliminating topmost occurrences of $\refl$ rule instances. \\

\textit{Base Case.} 
 An application of $\stitrefl$ on an initial sequent (
below left), where possibly $w=u$, can be re-written as an instance of the $\id$ rule (
below right): 
\begin{center}
\begin{tabular}{c @{\hskip 2 em} c}
\AxiomC{$\R, \R_{i}ww, u:p, u:\overline{p}, \Gamma$}
\RightLabel{$\stitrefl$}
\UnaryInfC{$\R, u:p, u:\overline{p}, \Gamma$}
\DisplayProof

&

\AxiomC{}
\RightLabel{$\id$}
\UnaryInfC{$\R, u:p, u:\overline{p}, \Gamma$}
\DisplayProof
\end{tabular}
\end{center}

\textit{Inductive step.} We show the inductive step for the $\ioa$ and $\choicerule$ rules, as well as for the non-trivial $\stitdiar$, $\stiteucl$, and $\propag$ cases. All other cases are resolved by applying IH to the premise followed by an application of the corresponding rule.\\

\noindent
(i) The $\stitrefl$ rule may be freely permuted with the $\ioa$ rule:

\begin{center}
\begin{tabular}{c @{\hskip 1em} c}
\AxiomC{$\R,\R_{i}ww, \R_{1}u_{1}v, ..., \R_{n}u_{n}v, \Gamma$}
\RightLabel{$\ioa$}
\UnaryInfC{$\R,\R_{i}ww, \Gamma$}
\RightLabel{$\stitrefl$}
\UnaryInfC{$\R,\Gamma$}
\DisplayProof

&

\AxiomC{$\R,\R_{i}ww, \R_{1}u_{1}v, ..., \R_{n}u_{n}v, \Gamma$}
\RightLabel{$\stitrefl$}
\UnaryInfC{$\R,\R_{1}u_{1}v, ..., \R_{n}u_{n}v,\Gamma$}
\RightLabel{$\ioa$}
\UnaryInfC{$\R,\Gamma$}
\DisplayProof
\end{tabular}
\end{center}

\noindent
(ii) We may easily permute the $\stitrefl$ rule with the $\choicerule$ rule:

\begin{center}
\begin{tabular}{c}
\AxiomC{$\Big\{ \R, \R_{i}uu, \R_{i}w_{k}w_{j}, \Gamma \ \Big| \ 0 \leq k \leq n-1 \text{, } k+1 \leq j \leq n \Big\}$}
\RightLabel{$\choicerule$}
\UnaryInfC{$\R,\R_{i}uu,\Gamma$}
\RightLabel{$\stitrefl$}
\UnaryInfC{$\R,\Gamma$}
\DisplayProof
\end{tabular}
\end{center}

\begin{center}
\begin{tabular}{c}
\AxiomC{$\Big\{ \R, \R_{i}uu, \R_{i}w_{k}w_{j}, \Gamma \ \Big| \ 0 \leq k \leq n-1 \text{, } k+1 \leq j \leq n \Big\}$}
\RightLabel{$\stitrefl \times \frac{n(n+1)}{2}$}
\UnaryInfC{$\Big\{ \R, \R_{i}w_{k}w_{j}, \Gamma \ \Big| \ 0 \leq k \leq n-1 \text{, } k+1 \leq j \leq n \Big\}$}
\RightLabel{$\choicerule$}
\UnaryInfC{$\R,\Gamma$}
\DisplayProof
\end{tabular}
\end{center}

\noindent
(iii) In the case of $\stitdiar$, when applying the $\stitrefl$ rule to the first premise of the left derivation instead, the propagation rule $\propag$ may be applied to the resulting sequent $\Lambda = \R, w:\agdia \phi, w:\phi, \Gamma$ since the empty string $\varepsilon \in \mathcal{P}_{\Lambda}(w,w) \cap L_{i}$:

\begin{center}
\begin{tabular}{c @{\hskip 2em} c}

\AxiomC{$\R, \R_{i}ww, w:\agdia \phi, w:\phi, \Gamma$}
\RightLabel{$\stitdiar$}
\UnaryInfC{$\R, \R_{i}ww, w:\agdia \phi, \Gamma$}
\RightLabel{$\stitrefl$}
\UnaryInfC{$\R, w:\agdia \phi, \Gamma$}
\DisplayProof

&

\AxiomC{$\R, \R_{i}ww, w:\agdia \phi, w:\phi, \Gamma$}
\RightLabel{$\stitrefl$}
\UnaryInfC{$\R, w:\agdia \phi, w:\phi, \Gamma$}
\RightLabel{$\propag$}
\UnaryInfC{$\R, w:\agdia \phi, \Gamma$}
\DisplayProof
\end{tabular}
\end{center}

\noindent
(iv) The non-trivial case of permuting the $\stitrefl$ rule with the $\stiteucl$ rule is shown below. The case may be resolved by leveraging admissibility of contraction. Note that dashed lines (below right) have been used to represent an application of height-preserving admissibility of contraction (Lem. \ref{lm:calculus-properties}), 
 the use of which decreases the height of the derivation by 1. As a consequence, 
  the height of the $\stitrefl$ rule application 
 also decreases by 1. The two cases are accordingly:

\begin{center}
\begin{tabular}{c @{\hskip 2em} c}
\AxiomC{$\R,\R_{i}ww, \R_{i}ww, \R_{i}ww, \Gamma$}
\RightLabel{$\stiteucl$}
\UnaryInfC{$\R, \R_{i}ww, \R_{i}ww, \Gamma$}
\RightLabel{$\stitrefl$}
\UnaryInfC{$\R, \R_{i}ww, \Gamma$}
\DisplayProof

&

\AxiomC{$\R,\R_{i}ww, \R_{i}ww, \R_{i}ww, \Gamma$}
\dashedLine
\RightLabel{Lem. \ref{lm:calculus-properties}}
\UnaryInfC{$\R, \R_{i}ww, \Gamma$}
\DisplayProof
\end{tabular}
\end{center}

\begin{center}
\begin{tabular}{c  @{\hskip 2em} c}
\AxiomC{$\R,\R_{i}ww, \R_{i}wu, \R_{i}wu, \Gamma$}
\RightLabel{$\stiteucl$}
\UnaryInfC{$\R, \R_{i}ww, \R_{i}wu, \Gamma$}
\RightLabel{$\stitrefl$}
\UnaryInfC{$\R, \R_{i}wu, \Gamma$}
\DisplayProof

&

\AxiomC{$\R,\R_{i}ww, \R_{i}wu, \R_{i}wu, \Gamma$}
\dashedLine
\RightLabel{Lem. \ref{lm:calculus-properties}}
\UnaryInfC{$\R, \R_{i}ww,\R_{i}wu, \Gamma$}\RightLabel{$\stitrefl$}
\UnaryInfC{$\R,\R_{i}wu,\Gamma$}
\DisplayProof
\end{tabular}
\end{center}

\noindent
(v) Last, the non-trivial case of permuting $\stitrefl$ over $\propag$ occurs when a relational atom of the form $\R_{i}ww$ is auxiliary in $\propag$. We therefore assume that in the first premise $\Lambda_{1} = \R, \R_{i}ww, x:\agdia \phi, y : \phi, \Gamma$ of the initial derivation (below left), the propagation rule is applied because there exists some string $\agdia^{n} \in \mathcal{P}_{\Lambda_{1}}(x,y) \cap L_{i}$ with $n \neq 0 \in \mathbb{N}$ and there exists a sequence of transitions from $x$ to $y$ containing transitions of the form $w \overset{\agdia}{\longrightarrow}w$ (
defined relative to $\R_{i}ww$). If instead we apply the $\stitrefl$ rule first (below right), then we obtain the sequent $\Lambda_{2} = \R, x:\agdia \phi, y:\phi, \Gamma$ which no longer contains the relational atom $\R_{i}ww$ that was used to apply the propagation rule in the initial derivation (left). Nevertheless, since $\agdia^{n} \in \mathcal{P}_{\Lambda_{1}}(x,y)$, there exists a sequence of transitions $x \overset{\agdia}{\longrightarrow} z_{1} \overset{\agdia}{\longrightarrow} \cdots z_{n-1} \overset{\agdia}{\longrightarrow} y$ containing transitions of the form $w \overset{\agdia}{\longrightarrow} w$. If we delete all transitions of the form $w \overset{\agdia}{\longrightarrow} w$ from the sequence of transitions, we will still have a valid sequence of transitions between $x$ and $y$. That is, if $\cdots z_i \overset{\agdia}{\longrightarrow} w \overset{\agdia}{\longrightarrow} w \overset{\agdia}{\longrightarrow} z_j \cdots$ occurs in our sequence, then clearly $\cdots z_i \overset{\agdia}{\longrightarrow} w \overset{\agdia}{\longrightarrow} z_j \cdots$ is a valid sequence.\footnote{Note that in the marginal case where our sequence of transitions is of the form $w \overset{\agdia}{\longrightarrow} w \overset{\agdia}{\longrightarrow} \cdots w \overset{\agdia}{\longrightarrow} w$, with all states (including $x$ and $y$) equal to $w$, the corresponding propagation automaton is $\mathcal{P}_{\Lambda_{2}}(w,w)$, which accepts the empty string $\varepsilon \in L_{i}$, thus allowing for the desired application of $\propag$.} Moreover, the resulting string $\agdia^{k}$ with $k < n$ will be in $L_{i}$. Therefore, since there exists some string $\agdia^{k} \in \mathcal{P}_{\Lambda_{2}}(x,y) \cap L_{i}$, the propagation rule may still be applied to $\Lambda_{2}$. 
\begin{center}
\begin{tabular}{c @{\hskip 3em} c}
\AxiomC{$\R, \R_{i}ww, x:\agdia \phi, y : \phi, \Gamma$}
\RightLabel{$\propag$}
\UnaryInfC{$\R, \R_{i}ww, x:\agdia \phi, \Gamma$}
\RightLabel{$\stitrefl$}
\UnaryInfC{$\R, x:\agdia \phi, \Gamma$}
\DisplayProof

&

\AxiomC{$\R, \R_{i}ww, x:\agdia \phi, y : \phi, \Gamma$}
\RightLabel{$\stitrefl$}
\UnaryInfC{$\R, x:\agdia \phi, y:\phi, \Gamma$}
\RightLabel{$\propag$}
\UnaryInfC{$\R, x:\agdia \phi, \Gamma$}
\DisplayProof
\end{tabular}
\end{center}

\end{proof}

\begin{customthm}{\ref{thm:soundness-ldmnl}}[Soundness of $\calc$]
Every sequent derivable in $\calc$ is valid.
\end{customthm}

\begin{proof} We know by Thm. \ref{thm:g3ldm-properties} that all rules of $\calc$, with the exception of $\propag$, preserve validity. We therefore only consider the $\propag$ case and argue by contraposition that if the conclusion of the rule is invalid, then so is the premise. Consider the following inference:

\begin{center}
\AxiomC{$\R, w:\agdia \phi, u:\phi, \Gamma$}
\RightLabel{$\propag$}
\UnaryInfC{$\R,w:\agdia \phi, \Gamma$}
\DisplayProof
\end{center}

Let $\Lambda_{1} = \R, w:\agdia \phi, u:\phi, \Gamma$ and $\Lambda_{2} = \R, w:\agdia \phi, \Gamma$. Assume that there exists a model $M$ with interpretation $I$ such that $\R,w:\agdia \phi, \Gamma$ is not satisfied in $M$ with $I$. In other words, $\R$ holds in $M$ with $I$, but for all labelled formulae $v:\psi$ in $w:\agdia \phi, \Gamma$ the following holds: $M,v^{I} \not\Vdash \psi$.

We additionally assume that the side condition holds; that is, there exists some string $\omega \in \mathcal{P}_{\Lambda_{1}}(w,u) \cap L_{i}$. Since $\omega \in L_{i}$, we know that $\omega$ is of the form $\agdia^{n}$ with $n \in \mathbb{N}$. Additionally, since $\omega = \agdia^{n} \in \mathcal{P}_{\Lambda_{1}}(w,u)$, this means that there exists a sequence of relational atoms of the form $\R_{i}wz_{1}, \cdots, \R_{i}z_{n-1}u$ or $\R_{i}uz_{1}, \cdots, \R_{i}z_{n-1}w$ connecting $w$ and $u$. Therefore, we have that $(w^{I},z_{1}^{I}),$ $\cdots,$ $(z_{n-1}^{I},u^{I}) \in \R_{i}$ or $(u^{I},z_{1}^{I}), \cdots, (z_{n-1}^{I},w^{I}) \in \R_{i}$ in $M$ under $I$, which implies that $w$ and $u$ are mapped to worlds in the same equivalence class, \textit{i.e.}, $\R_{i}w^{I}u^{I}$ holds without loss of generality. By assumption, $M,w^{I} \not\Vdash \agdia \phi$ holds, meaning that $M,w^{I} \Vdash \agbox \overline{\phi}$ holds. Since both $R_{i}w^{I}u^{I}$ and $M,w^{I} \Vdash \agbox \overline{\phi}$ hold, we know that $M,u^{I} \Vdash \overline{\phi}$, implying that $M,u^{I} \not\Vdash \phi$. Thus, the premise is falsified by $M$ under $I$.

\end{proof}

\begin{customlem}{\ref{lm:forestlike-invariance}} Every labelled sequent generated throughout the course of computing $\mathtt{Prove_{n}}(w:\phi)$ is forestlike.

\end{customlem}

\begin{proof} We prove the result by induction on the number of instructions executed. Note that the input sequent $w : \phi$ is trivially forestlike.

\textit{Base Case.} We assume that one of the instructions (2 - 7) has been executed in the algorithm (initially, instructions 1 and 8 cannot be executed): (2) If $w : \phi$ is stable, then no sequent is generated. (3) If instruction 3 is executed, then in case (i) $w : \phi = w : \psi \vee \chi$, so the generated sequent is $w : \phi, w:\psi, w:\chi$, which is forestlike; in case (ii) $w : \phi = w : \psi \wedge \chi$, so the sequents $w: \phi, w:\psi$ and $w:\phi, w:\chi$ are generated, which are both forestlike. (4) If instruction 4 is executed, then $w:\phi = w:\agdiaone \psi$, so the sequent generated is of the form $w:\phi, w:\psi$, which is forestlike. (5) If instruction 5 is executed, then $w:\phi = w:\Diamond \psi$, so the sequent generated is of the form $w:\phi, w:\psi$, which is forestlike. (6) If instruction 6 is executed, then $w : \phi = w:[1] \psi$, and the sequent $\R_{1}wv,w_\phi, v:\psi$ is generated, which is forestlike. (7) If instruction 7 is executed, then $w : \phi = w : \Box \psi$, so the sequent $w:\phi, v:\psi$ is generated, which is forestlike.

\textit{Inductive Step.} We assume that our input sequent is forestlike, and argue that the generated sequent is forestlike: (1) If instruction 1 is executed, then no sequent is generated. (2) If instruction 2 is executed, then no sequent is generated. (3) Each of the operations in instruction 3 preserves the set of relational atoms $\R$ as well as the set of labels in the sequent; hence, the sequent generated after an execution of instruction 3 will be forestlike. (4) Instruction 4 labels nodes in the graph of the input sequent with additional formulae, but does not change the structure of the graph; therefore, the generated (output) sequent will be forestlike. (5) Similar to case 4. (6) Instruction 6 adds one additional $\R_{1}$ edge to a fresh labelled node $v$ in the graph of the input sequent, which adds additional branching in the graph of the generated sequent, and thus preserves the forestlike structure of the sequent. (7) Instruction 7 adds a new labelled formula to the sequent, which is akin to adding a new, disjoint labelled node to the graph of the input sequent; this preserves the forestlike structure of the sequent. (8) Instruction 8 connects the root of one choice-tree to the root of another choice-tree in the graph of the input sequent; the result is another tree, and thus, the generated sequent will be forestlike.

\end{proof}

\begin{customthm}{\ref{thm:termination-prove}}[Termination] For every labelled formula $w:\phi$, $\mathtt{Prove_{n}(w:\phi)}$ terminates.

\end{customthm}

\begin{proof} Let $\textit{sufo}(\phi)$ be the multiset of subformulae of $\phi$ defined in the usual way. Observe that new labels are only created through instructions 6 and 7 of the algorithm, and each time instruction 6 or 7 executes, the formula $w:[1] \psi$ or $w:\Box \psi$ (resp.) responsible for the instruction's execution, no longer influences the non-$[1]$-realization or non-$\Box$-realization of $w$. Therefore, the number of labels in any sequent generated by the algorithm is bounded by the number of $[1]$- and $\Box$-formulae contained in $\textit{sufo}(\phi)$ plus 1 (which is the label of the input formula). Moreover, the number of $\R_{1}$ relational atoms is bounded by the number of $[1]$-formulae.


Instructions 3--5 add strict subformulae of formulae occurring in $\textit{sufo}(\phi)$ and do not create new labels or relational atoms. Due to the blocking conditions in instructions 3--5, any formula that is added with a label will only be added once. Therefore, the number of executions of instructions 3--5 is bounded by $|\textit{sufo}(\phi)|$ multiplied by 1 plus the number of $[1]$- and $\Box$-formulae occurring in $\textit{sufo}(\phi)$.

Last, observe that instruction 6 increases the breadth or height of a choice-tree in the input sequent, whereas instruction 7 adds a new label which acts as the root of a new choice-tree. This implies that the number of choice-trees is bounded by the number of $\Box$-formulae occurring in $\textit{sufo}(\phi)$. Each time instruction 8 executes the number of choice-trees in a resulting sequent decreases by 1. Since the number of choice-trees is bounded by the number of $\Box$-subformulae, eventually the number of choice-trees will decrease to $k \leq n$, at which point, the generated sequent will be $n$-choice consistent, and instruction 8 will no longer be applicable.

 Therefore, the algorithm will terminate.

\end{proof}



\end{document}